\documentclass[aps,pra,onecolumn,superscriptaddress,nofootinbib,longbibliography,floatfix]{revtex4-2}

\usepackage[T1]{fontenc}
\usepackage{lmodern}
\usepackage{amsmath,amssymb,mathtools,bm}
\usepackage{booktabs,float}
\usepackage{graphicx}
\usepackage[expansion=false]{microtype}
\usepackage[colorlinks=true,allcolors=blue]{hyperref}

\newcommand{\Diss}[1]{\mathcal D[#1]}
\newcommand{\tr}{\operatorname{tr}}
\newcommand{\cone}{\operatorname{cone}}
\newcommand{\spann}{\operatorname{span}}
\newcommand{\Ad}{\operatorname{Ad}}
\newcommand{\Stab}{\operatorname{Stab}}
\newcommand{\Herm}{\operatorname{Herm}}

\newcommand{\F}{\mathcal F}
\newcommand{\T}{\mathcal T}
\newcommand{\X}{\mathcal X}
\newcommand{\G}{\mathcal G}

\newcommand{\eps}{\varepsilon}
\newcommand{\dtr}{d_{\rm tr}}
\newcommand{\ddia}{d_{\diamond}}

\newtheorem{theorem}{Theorem}
\newtheorem{proposition}[theorem]{Proposition}
\newtheorem{lemma}[theorem]{Lemma}
\newtheorem{corollary}[theorem]{Corollary}
\newtheorem{definition}[theorem]{Definition}
\newtheorem{remark}[theorem]{Remark}
\newtheorem{example}[theorem]{Example}
\newenvironment{proof}{\par\noindent\textit{Proof.}\ }{\hfill$\square$\par}

\hypersetup{
  pdftitle={From Static Lindblad Stabilizability to Robust Physical Control Sequences},
  pdfauthor={Yibin Wang and Ruifeng Liu},
  pdfsubject={Open quantum systems and quantum control}
}

\makeatletter
\renewcommand\subsection{%
  \@startsection{subsection}{2}{\z@}%
    {.5cm \@plus .5ex \@minus .2ex}%
    {.20cm}%
    {\normalfont\small\bfseries\centering}%
}
\makeatother

\begin{document}

\title{From Static Lindblad Stabilizability to Robust Physical Control Sequences}

\author{Yibin Wang}
\email{yibinw0210@gmail.com}
\affiliation{Graduate School of Mathematics, Nagoya University, Nagoya, 464-8601, Aichi, Japan.}
\author{Ruifeng Liu}
\email{ruifeng.liu@mail.xjtu.edu.cn}
\affiliation{Ministry of Education Key Laboratory for Nonequilibrium Synthesis and Modulation of Condensed Matter, Shaanxi Province Key Laboratory of Quantum Information and Quantum Optoelectronic Devices, School of Physics, Xi'an Jiaotong University, Xi'an 710049, China.}

\begin{abstract}
We derive checkable convergence, error, and execution-cost bounds for
repeating a supplied finite control sequence from a static Lindblad
stabilization certificate and complete implementation data, including
concurrent background dynamics. The static decay bound and error bounds from
static to planned evolution and from the plan to execution together bound the
actual cycle's contraction factor in one fixed Lyapunov metric. For
target-preserving planned and actual paths and a controller returning to its
initial state each cycle, a uniform contraction bound below one gives
cycle-end convergence from every initial state under cycle-varying
uncertainty. Uniformly bounded within-cycle amplification in that metric
extends this convergence throughout physical time when cycle durations are
bounded above and away from zero. With leakage, uniform error bounds relative
to certified target-preserving paths give explicit long-time distance bounds.
Cycle and within-cycle bounds also screen a specified finite hardware library
for sequences certified to meet and maintain a prescribed tolerance. In a
four-level cascade with finite pulses, calibration and timing errors,
always-on dephasing, and bounded pumping, the selected sequence brings every
initial state within trace distance $0.1$ of the target in one complete cycle
of explicitly bounded duration and maintains that tolerance under repetition.
\end{abstract}

\maketitle

\section{Introduction}
\label{sec:introduction}

Reservoir engineering asks for dynamics that drive every initial state toward
a prescribed quantum state or code space. Laser-controlled reservoir design
was proposed as a way to select a motional master equation
\cite{poyatos1996reservoir}, and designed dissipation has been developed for
steady-state quantum computation and correlated-state preparation
\cite{verstraete2009dissipation}.  Experiments have combined coherent gates
with optical pumping to realize dissipative entangled-state preparation in
trapped ions \cite{barreiro2011open}, and continuous drives with an engineered
reservoir to autonomously stabilize a Bell state in superconducting circuits
\cite{shankar2013autonomous}.

For a time-independent master equation, static stabilization asks whether the
available cone of candidate generator directions contains a
Gorini--Kossakowski--Lindblad--Sudarshan (GKLS) generator with stable
transverse dynamics and an invariant \emph{target face} consisting of the
density operators supported on the target subspace.
The cone does not specify the rate assigned to a ray, whether a sum
is directly callable, or the path and cost of a conjugating pulse.
We ask how a static stabilization certificate and complete implementation
data can certify a finite control sequence's convergence, intermediate errors,
and physical cost.

We call one complete finite-duration hold, pulse, transition, dead interval,
or cooldown a \emph{primitive}, whose generator contains every simultaneous
always-on term. A \emph{physical word} is a finite chronological sequence of
these primitives, and a \emph{prefix} is its evolution from the start of a
cycle to any intermediate time.

\subsection{Main results}

Assume that the complete paths preserve the target face, the controller returns
to the same reproducible boundary state after each cycle, and cycle durations
are bounded above and away from zero. Under these conditions, when the
derived cycle-contraction factor is uniformly below one and the prefixes are
uniformly bounded, the primitive-certified compiler
(Theorem~\ref{thm:robust-path-compiler}) certifies all-time attraction to the
target face under repeated execution of a supplied finite physical word, for
every initial state and cycle-varying uncertainty. The compiler
derives these checkable bounds from a static stabilization certificate and
complete implementation data, combining static decay with nominal-synthesis
and complete-path errors in one fixed Lyapunov metric.
Nominal synthesis compares a reference word with the static evolution,
while complete-path bounds compare actual operations with their reference
paths, including duration \mbox{uncertainty} and amplification by surrounding
operations. A common-ground qutrit realizes the fixed-metric construction
with bounded directed pulses, uncertain decay gains, and known always-on
dephasing.

For actual paths with uniformly bounded errors relative to target-preserving
shadow words satisfying these compiler conditions,
Theorem~\ref{thm:primitive-practical-tube} gives an explicit practical tube.
When the cycle-and-prefix bounds certify a prescribed tolerance, both branches
give worst-case execution-time and resource guarantees.

Two qubit obstructions show why a finite generator decomposition or an
algebraic orbit identity needs additional path-compatibility data. In the first, a two-label presentation without resets or
auxiliary operations has a stabilizing average, but no nonempty finite
positive-dwell word over its two labels fixes the pure target at a cycle
boundary (Theorem~\ref{thm:static-word-no-go}). In the second, a stabilizing
unitary-orbit atom has a finite-cost realization at cycle boundaries, but no
word of positive physical duration preserves the pure target throughout
(Theorem~\ref{thm:orbit-word-no-go}). The bounded conjugating pulse makes a full
excursion from the target. The examples distinguish generator-level
stabilization, cycle-boundary stabilization, and stabilization throughout
physical execution.

A common acceptance test (Proposition~\ref{prop:finite-search-test}) uses
cycle-progress and prefix bounds from either a fixed-metric certificate or a
positive-observable certificate to certify a prescribed tolerance at all
times after preparation and supports certified selection of words from a
declared finite hardware library. In a four-level cascade, the
survival-observable certificate covers finite pulse and duration errors,
always-on dephasing, and bounded pumping out of the target. Over the declared
joint error box, a five-command sequence selected from a library of 1650
parameterized words guarantees, for every initial state, trace distance at
most $0.1$ at the end of a preparation cycle finishing by time $14.3107$ and
throughout continued execution. Times are measured in inverse nominal
cascade-rate units.

\subsection{Relation to previous stabilization results}

Common-target switching and dissipative feedback provide the nearest
stabilization results \cite{scaramuzza2015switching,liang2024feedback}.
Table~\ref{tab:nearest-results} compares their assumptions and conclusions
with the complete-path problem. For noninvariant unmodulated generators, a separate strict descent
condition and state-dependent gains yield asymptotic stabilization in mean and
almost surely \cite{liang2024feedback}. The state information and gain modulation distinguish this
feedback construction from our fixed open-loop words with unavoidable excitation. Generator
decompositions yield sequential
simulations when the constituent semigroups and arbitrary unitaries are already
implementable \cite{sweke2015universal}.  Positive product formulas for open
systems \cite{kliesch2011churchturing,childs2017sparse,iserles2024splitting},
bounded Eulerian controls \cite{viola2003bounded}, soft-pulse Markovian models
\cite{pryadko2009soft}, and finite-dimensional Duhamel perturbation estimates
\cite{szehr2013perturbation} provide standard analytic ingredients.  The leakage
estimate is related to input-to-state practical stability \cite{sontag1995iss},
and we prove the discrete cycle-and-prefix recursion used for that estimate.

Robust Lyapunov theory gives smooth stability characterizations for bounded
time-varying parameter families \cite{lin1996robust_lyapunov}. The present
finite input specifies one quadratic metric.  Robustness and error
certification in open-system control and simulation have also been studied
through short-pulse dynamical
decoupling \cite{gong2013dd}, stability of local observables for rapidly mixing
local Lindbladians \cite{cubitt2015stability}, and structured Bloch-generator
uncertainty \cite{schirmer2022robust}.

\begin{table}[H]
\centering
\caption{Related results. Added guarantees are uniform over the declared
complete-path uncertainty set.}
\label{tab:nearest-results}
\small
\begin{tabular}{@{}l@{\hspace{0.02\linewidth}}l@{\hspace{0.02\linewidth}}l@{}}
\toprule
\parbox[c]{0.20\linewidth}{\raggedright\strut \textbf{Work}\strut} &
\parbox[c]{0.45\linewidth}{\raggedright\strut \textbf{Conditions and result}\strut} &
\parbox[c]{0.31\linewidth}{\raggedright\strut \textbf{Added guarantees}\strut}\tabularnewline
\midrule
\parbox[c]{0.20\linewidth}{\raggedright\strut  Scaramuzza--Ticozzi \cite{scaramuzza2015switching}\strut} &
\parbox[c]{0.45\linewidth}{\raggedright\strut  Common invariant target and stable convex combination: stabilization for
sufficiently short switching periods, with within-period bounds.\strut} &
\parbox[c]{0.31\linewidth}{\raggedright\strut  Joint-error bounds for finite pulses, idle paths, and physical durations.\strut}\tabularnewline\addlinespace[1.25pt]
\parbox[c]{0.20\linewidth}{\raggedright\strut  Liang et al.\ \cite{liang2024feedback}\strut} &
\parbox[c]{0.45\linewidth}{\raggedright\strut  Invariant target and uniform exponential descent in one Lyapunov observable
under statewise mode choice: dwell bounds and exponential stabilization
in mean and almost surely.\strut} &
\parbox[c]{0.31\linewidth}{\raggedright\strut  Certified fixed open-loop words.\strut}\tabularnewline\addlinespace[1.25pt]
\parbox[c]{0.20\linewidth}{\raggedright\strut  Szehr--Wolf \cite{szehr2013perturbation}\strut} &
\parbox[c]{0.45\linewidth}{\raggedright\strut  Unique stationary state and exponential reference convergence: finite-time
and asymptotic bounds from channel or generator perturbation norms.\strut} &
\parbox[c]{0.31\linewidth}{\raggedright\strut  Complete-primitive perturbations through uncertain cycles and prefixes
for target faces.\strut}\tabularnewline\addlinespace[1.25pt]
\parbox[c]{0.20\linewidth}{\raggedright\strut  Kraft et al.\ \cite{kraft2026bounded_error}\strut} &
\parbox[c]{0.45\linewidth}{\raggedright\strut  Learned statistical uncertainty in Hamiltonian and Lindbladian parameters:
observable-error bounds.\strut} &
\parbox[c]{0.31\linewidth}{\raggedright\strut  Target preparation and physical execution-time guarantees throughout operation
over a deterministic uncertainty set.\strut}\tabularnewline
\bottomrule
\end{tabular}
\end{table}

Sections~\ref{sec:presentations}--\ref{sec:static} specify the physical model,
target semantics, static input, averaging obstruction, and compatibility tests.
Section~\ref{sec:physical-words} develops robust path certification and leakage
bounds. Sections~\ref{sec:qutrits} and~\ref{sec:certified-search} give the qutrit
and four-level search examples. The appendices supply static algebraic
certificates, direct-label schedules, path calculations, and an exactly
solvable cascade with matching presentation-relative resource bounds.

\section{Calibrated presentations and target semantics}
\label{sec:presentations}

\subsection{Target faces and transverse coordinates}

Let $\mathcal H=\mathbb C^d$ and let $\Herm(\mathcal H)$ be the real vector
space of Hermitian operators. We write $\rho$ for a density operator on
$\mathcal H$, with $\rho\succeq0$ and $\tr\rho=1$. For an operator $J$, write
\begin{equation}
 \Diss{J}(\rho)=J\rho J^\dagger-\tfrac12\{J^\dagger J,\rho\}.
\label{eq:dissipator}
\end{equation}
A GKLS generator with Hamiltonian $H=H^\dagger$ and jump operators $V_a$
has the form~\cite{gorini1976completely,lindblad1976generators}
\begin{equation}
 \mathcal L(\rho)=-i[H,\rho]+\sum_a\Diss{V_a}(\rho).
\label{eq:gkls}
\end{equation}
Its exponential is a completely positive trace-preserving (CPTP) map.  In a
time-dependent segment, the Hamiltonian,
jump operators, and any always-on terms in Eq.~\eqref{eq:gkls} are the complete
dynamics at that operating point.

\begin{definition}[Targets and attraction]
\label{def:target}
Let $P$ be an orthogonal projector of rank $1\le r<d$, set $Q=I-P$, and define
\begin{equation}
 \T_P=\{X\in\Herm(\mathcal H):X=PXP\},\qquad \F_P=\{\rho\in\T_P:\rho\succeq0,\ \tr\rho=1\}.
\label{eq:target-face}
\end{equation}
The projector $P$ is the target projector and $\F_P$ is its target face.  For
$r=1$, the unique state in the face is the target state.  The face is
\emph{globally asymptotically stable} (GAS) when it is invariant and every
initial density operator approaches it in trace distance.
\end{definition}

For $r>1$, GAS means attraction to the whole face and does not by itself
select one mixed state inside it. We use
\begin{equation}
 \dtr(\rho,\F_P)=\inf_{\sigma\in\F_P}\tfrac12\|\rho-\sigma\|_1.
\label{eq:trace-distance-face}
\end{equation}
For channels, the half diamond distance is
$\ddia(\Phi,\Psi)=\tfrac12\|\Phi-\Psi\|_\diamond$.  Primitive leakage bounds
below use the full superoperator diamond norm $\|\cdot\|_\diamond$, with the factor
$1/2$ inserted only when a channel distance is explicitly denoted by
$\ddia$.

The quotient
$\X_P=\Herm(\mathcal H)/\T_P$ separates approach to the target face from
motion within it by retaining only transverse displacement.
We identify this quotient isometrically with the Hilbert--Schmidt orthogonal complement
$\T_P^\perp$ by assigning $[X]$ the representative $Q_X=X-PXP$.  Once
$\T_P$ is invariant, a generator induces a \emph{transverse matrix} on a
Hilbert--Schmidt orthonormal real Hermitian basis of $\T_P^\perp$. A channel
similarly induces a transverse map. For Hermitian $X$, write $z(X)$ for the
coordinates of $Q_X$ in this basis, so that $\|z(X)\|_2=\|Q_X\|_{\rm HS}$.

\subsection{Calibrated labels and complete operations}

The physical input records both the available dynamics and the controller
states needed to concatenate them.

\begin{definition}[Calibrated physical presentation]
\label{def:presentation}
A calibrated presentation is a finite collection of primitive labels and a
finite directed graph of controller states.  A label identifies one complete
operating point: its generator or channel path, physical rate and amplitude
range, admissible duration, concurrency rule, and cost vector.  A graph vertex
records every controller variable needed to repeat the experiment, including
the laboratory frame, phase when relevant, actuator configuration,
and environment/reset state.  A directed edge records a
complete transition path and its duration, always-on dynamics, calibration,
and cost.  A reverse edge is separate data.  Algebraically equal
superoperators remain different labels when their physical invocations differ.
\end{definition}

A complete generator for a calibrated label includes the commanded
Hamiltonian, raw dissipation, simultaneous compensation, cross-talk, and all
always-on terms. A calibration relation
$\widetilde G=rG$ with $r>0$ is tied to one operating point and cannot be used
to rescale an independent drift.

\begin{definition}[Direct orbit atom]
\label{def:orbit-atom}
A direct orbit atom is a complete calibrated generator $G$ supplied as one
label.  When orbit membership is used, the label comes with the algebraic
identity
$G=\Ad_U G_0:=\mathcal U\circ G_0\circ\mathcal U^{-1}$.  Here $G_0$ is the declared
reference generator in the orbit model and
$\mathcal U(\rho)=U\rho U^\dagger$ is the channel of a declared unitary orbit
element $U$.
\end{definition}

The orbit identity relates generators algebraically, but implementing
$\mathcal U$ or $\mathcal U^{-1}$ still requires separate availability and
cost data in the physical presentation.

\subsection{Physical words and target semantics}

\begin{definition}[Physical word and cold boundary]
\label{def:physical-word}
A physical word is a finite closed directed walk based at a declared
\emph{cold boundary}, a reproducible controller state at which periods may be
concatenated. Under the declared system-channel factorization, its channel is
the chronological product of the complete hold and edge channels along the walk.
\end{definition}

If system--environment memory survives a cut, the system segment maps need
not compose as CPTP channels, and the dynamical model must account for that
memory.

A \emph{physical prefix} is the channel accumulated from the cold boundary to
any intermediate time in a word.  \emph{Micromotion} is the within-period
evolution seen through these prefixes but not through the full period map.

Four target semantics will recur.  \emph{All-time invariance} means every
physical prefix maps $\F_P$ into itself.  \emph{Stroboscopic invariance} asks
only that the full period map preserve the face at the chosen cycle boundary.
For a rank-one target, an attractive stroboscopic map may generate a
nonconstant \emph{periodic orbit} within a cycle.  A \emph{practical tube} is
only a uniform upper bound on the distance to the target.

\subsection{Execution costs}

The presentation assigns each complete invocation or operating setting a
\emph{physical signature}: a discrete identifier for its prescribed command,
concurrent dynamics, and transition direction when applicable. A signature
covers the complete prescribed waveform and its admitted uncertainty family.
Every graph vertex records its physical signature. Two signatures are equal
exactly when the presentation identifies their declared invocation data for
switch counting. The idle setting, denoted $\mathrm{IDLE}$, and distinct
directed pulse commands have separate signatures.

A \emph{primitive call} is one issued active label or one transition that the
presentation declares countable as a call.  A \emph{dwell} is the time held at
an operating point.  \emph{Dead time} is an interval in which a scheduled
control is unavailable but complete dynamics still acts. \emph{Cooldown} is
a required delay before another call.  Always-on dynamics acts whether or not
the nominal primitive is active.  Wall time is physical elapsed time.  The
switch count is the number of changes between consecutive physical signatures,
including entry to and return from a cyclic boundary. When reported, the \emph{integrated control action} sums the amplitude norm over the
multiset
$\mathcal E_{\rm ctrl}$ of control paths in the word:
\begin{equation}
 \mathcal A_{\rm ctrl}
 =\sum_{e\in\mathcal E_{\rm ctrl}}\int_e\|u_e(s)\|_{\rm ctrl}\,ds,
\label{eq:integrated-control-action}
\end{equation}
where $u_e(s)$ is the control-amplitude vector along path $e$ and the
presentation fixes the control-amplitude norm. Thus the two swap
pulses in Sec.~\ref{subsec:common-ground} both contribute to this quantity.

\subsection{Uncertainty and reproducible closure}

\begin{samepage}
A \emph{word grammar} specifies allowed finite chronological primitive templates
and their parameter ranges. The exposure parameter $h$ labels the nominal
synthesis and need not equal the physical cycle duration.

\begin{definition}[Uncertain word family]
\label{def:uncertain-word-family}
Fix the Hilbert-space dimension $d$, target rank $r$, number of chronological
primitives $N$, and a nonempty compact exposure domain $H\subset(0,\infty)$.
An uncertain word family consists of the following data and conditions.
\begin{enumerate}
\item \emph{Reference and actual paths.}
For every $h\in H$, a finite word grammar supplies $N$ complete reference paths
and a set $\Omega(h)$ of complete actual paths with matching ordered types.
An empty $\Omega(h)$ gives no executable realization at $h$. Scalar coordinates
lie in closed bounded intervals. Normalized waveforms are bounded, equi-Lipschitz
on a fixed finite partition, and
piecewise $W^{1,\infty}$ there. The joint family
$\{(h,\omega):h\in H,\ \omega\in\Omega(h)\}$ is compact in the product of the
Euclidean scalar topology and the uniform path topology. Each uncertainty
coordinate is declared as an absolute physical width or as a relative width
with respect to a nonzero reference value.
\item \emph{Physical path data.}
Each path specifies its physical duration, amplitude and slew bounds, dead
time, cooldown, always-on dynamics, and every controller field at both
endpoints. Every nontrivial segment has duration in a closed interval bounded
away from zero. A zero-duration slot is allowed only when its system channel
and every controller field are exactly unchanged.
\item \emph{Cycle duration and closure.}
The declared total cycle duration obeys
\begin{equation}
 0<T_-\le T_{\rm cyc}(h,\omega)\le T_+<\infty
\label{eq:nonzeno-cycle-duration}
\end{equation}
throughout the joint family. Every actual word closes at the same cold
boundary field by field.
\end{enumerate}
\end{definition}
\end{samepage}

The positive lower bound on cycle duration prevents infinitely many cycles
from occurring in finite time (the non-Zeno property). The uncertainty may be
fixed but unknown, reset to an independent value at each cycle, or chosen
adversarially at each cycle. Generator and channel error bounds are expressed
in absolute norms in the stated units.

If perturbed controller fields do not close at the declared boundary, they must
be included as part of an enlarged dynamical state.  A nominal reverse command
alone does not prove closure: for planar rotations,
$R(-\theta+\epsilon_-)R(\theta+\epsilon_+)=
R(\epsilon_-+\epsilon_+)$.

\section{Static input and target-compatible controls}
\label{sec:static}

The physical construction takes as input a complete GKLS generator $L$ with
an invariant, attractive target face. We first extract its transverse decay
estimate, then identify the compatibility conditions needed to use its atoms
as controls. Appendix~\ref{app:static-proofs} gives a finite algebraic
existence certificate for algebraic input with compact semialgebraic orbit
data and a positive gauge. Calibration enters through the physical
presentation, and Theorem~\ref{thm:presentation-cost} shows that presentations
with the same wedge, and even the same untimed semigroup, can have different
physical costs.

\subsection{The static Lyapunov input}
Let $\Pi_P(X)=PXP$.  The no-leakage equation
\begin{equation}
 (I-\Pi_P)L\Pi_P=0
\label{eq:static-no-leakage}
\end{equation}
makes $\T_P$ invariant, so the transverse matrix $A_L$ on $\X_P$ is well
defined. Invariance thus precedes the transverse decay test below.

\begin{proposition}[Invariant-face stability]
\label{prop:fixed-generator-gas}
The face $\F_P$ is invariant and GAS under $e^{tL}$ if and only if
Eq.~\eqref{eq:static-no-leakage} holds and $A_L$ is Hurwitz, meaning that every
eigenvalue has strictly negative real part. The Hurwitz condition is equivalent
to positivity of the leading Routh--Hurwitz determinants, the standard
polynomial stability determinants formed from the coefficients of the monic
characteristic polynomial $\det(zI-A_L)$ in descending powers of $z$.
\end{proposition}

For a fixed dissipator with time-independent Hamiltonian control, invariance
and attraction are likewise separate: leakage blocks must vanish, while
attraction requires that the complementary subspace not remain invariant
\cite{ticozzi2012hamiltonian}.  To control nonnormal transients we retain a
strict Lyapunov metric. With the basis of the transverse quotient fixed once
and for all, define for $M\succ0$
\begin{equation}
 \|x\|_M=(x^\dagger Mx)^{1/2},\qquad
 \|A\|_M=\|M^{1/2}AM^{-1/2}\|_2.
\label{eq:M-induced-norm}
\end{equation}
For a quotient trajectory, write
$x(t)=e^{tA_L}x(0)$.  If
\begin{equation}
 m_-I\prec M\prec m_+I,\qquad
 A_L^{\mathsf T}M+MA_L\preceq-2\nu M-q_0I
\label{eq:static-lyapunov}
\end{equation}
with positive algebraic $m_-,m_+,\nu,q_0$, then
\begin{equation}
 \begin{aligned}
 &\|e^{tA_L}\|_M\le e^{-\nu t},\\
 &\|x(t)\|_2\le\sqrt{m_+/m_-}\,e^{-\nu t}\|x(0)\|_2,\\
 &\dtr(e^{tL}\rho,\F_P)\le2\sqrt{d\,m_+/m_-}\,e^{-\nu t}.
 \end{aligned}
\label{eq:conditioned-static-bound}
\end{equation}
The first inequality follows from
$d\|x(t)\|_M^2/dt\le-2\nu\|x(t)\|_M^2$. A change of coordinates acts on $M$
by congruence, making the induced norm covariant, whereas a raw Euclidean
norm need not be. The metric compiler uses this same norm for its primitive
and cycle comparisons. The explicit metric and prefactor
account for nonnormal dynamics in the uniform endpoint estimate, which
supplies the reference contraction used to assess the physical word in
Sec.~\ref{sec:physical-words}.

\subsection{A static average need not yield a target-preserving word}
A static decomposition may contain atoms that fail to preserve the target
individually. To exhibit this obstruction, consider a direct presentation in
which complete generators are callable labels and the coherent edge---the real
space of directly available Hamiltonian commutators, distinct from controller
graph transitions---is zero. Hamiltonian terms may still occur inside complete
labels. For the static certificate, the compact semialgebraic input and
positive-gauge functional (zero on the coherent edge, uniformly positive on
dissipative atoms) are specified in Appendix~\ref{app:static-proofs} and
Definition~\ref{def:positive-gauge}.

\begin{theorem}[Two-label qubit obstruction]
\label{thm:static-word-no-go}
There is a two-label compact positive-gauge direct presentation, with zero
coherent edge and no reset or auxiliary interstitial operation, whose average
generator has a finite certificate of Theorem~\ref{thm:finite-static-certificate}
for a rank-one GAS target, while no
nonempty finite positive-dwell word over its two labels fixes that target at a
cycle boundary.
\end{theorem}

Let $J=|0\rangle\langle1|$.  Write $X$, $Y$, and $Z$ for the standard Pauli
matrices in the ordered basis $\{|0\rangle,|1\rangle\}$, and set
$\mathcal D=\Diss{J}$,
$\mathcal H(\rho)=-i[X/2,\rho]$, and
\begin{equation}
 \mathcal L_\pm=\mathcal D\pm\mathcal H.
\label{eq:qubit-average-atoms}
\end{equation}
Their average is $\mathcal D$, which is GAS to
$P_0=|0\rangle\langle0|$, but
$\mathcal L_\pm(P_0)=\mp Y/2$. With \mbox{$\rho=(I+xX+yY+zZ)/2$} in Bloch
coordinates, define the purity defect
$\delta=1-(x^2+y^2+z^2)$. Along either atom,
\begin{equation}
 \dot\delta+\delta=(1-z)^2.
\label{eq:purity-defect}
\end{equation}
The first positive dwell therefore sends $P_0$ into the full-rank interior.
Every subsequent finite-time GKLS propagator preserves strict positivity, so
the word cannot return to the pure target.  Appendix~\ref{app:target-proofs}
gives the calculation and the strict-positivity lemma.

Fast switching may keep the trajectory in a shrinking tube, but requiring the
target to remain fixed at every physical time forces each used atom to pass an
independent compatibility test.  For $S=\{\mathcal L_+,\mathcal L_-\}$, the example realizes the strict inclusion
\begin{equation}
 \cone(S)\cap\ker N_P
 \supsetneq \cone(S\cap\ker N_P),
\label{eq:compatible-cone-gap}
\end{equation}
where $N_P(G)=(I-\Pi_P)G\Pi_P$.  Supplying the average itself as a primitive,
allowing a parallel compound label, adding a compensator, or permitting a
noninvertible reset defines another presentation.

This obstruction is minimal within zero-coherent-edge direct presentations:
dimension one has no nontrivial target, and with a single positive atom
$L=\lambda G$, invariance of $L$ forces invariance of $G$.

\subsection{Compatibility of callable generators}
The compatibility test applies to each complete generator actually invoked,
including its simultaneous compensation and background dynamics.

\begin{lemma}[GKLS target-face criterion]
\label{lem:gkls-face}
For a complete generator of the form in Eq.~\eqref{eq:gkls}, the face $\F_P$
is invariant if and only if
\begin{equation}
 QV_aP=0\ \text{for every }a,\qquad
 Q\left(H-\frac{i}{2}\sum_aV_a^\dagger V_a\right)P=0.
\label{eq:gkls-face-criterion}
\end{equation}
\end{lemma}

The positive $QQ$ block generated by jump leakage cannot be canceled by a
Hamiltonian.  Once that block vanishes, the $QP$ block gives the second equality.
For a fixed dissipator and background Hamiltonian $H_0$, executable
compensation requires a calibrated, costed simultaneous label for the complete
compensated generator, with Hamiltonian correction $H_c$ satisfying
\begin{equation}
 QH_cP=\frac{i}{2}Q\left(\sum_aV_a^\dagger V_a\right)P-QH_0P.
\label{eq:compensator}
\end{equation}
The Hamiltonian correction equation is a real-linear feasibility test.
A proof of Lemma~\ref{lem:gkls-face} appears in
Appendix~\ref{app:target-proofs}.

\begin{definition}[Compatible direct realization]
\label{def:direct-realization}
Let $L$ be a statically certified generator.  A target-compatible direct
realization is a finite family
$\{(G_j,\widetilde G_j,\lambda_j,r_j)\}_{j=1}^m$ such that
\begin{equation}
 \widetilde G_j=r_jG_j,\qquad r_j>0,\quad \lambda_j>0,\qquad L=\sum_{j=1}^m\lambda_jG_j,\qquad (I-\Pi_P)G_j\Pi_P=0.
\label{eq:direct-realization}
\end{equation}
Here $\widetilde G_j$ is the complete calibrated label used during the
physical dwell, $G_j=\widetilde G_j/r_j$, and $A_j=(G_j)_\perp$ is its induced
transverse matrix.  The calibrated transverse matrix is
$\widetilde A_j=r_jA_j$, and $A_L=\sum_j\lambda_jA_j$.
\end{definition}

A finite static decomposition becomes a compatible direct realization only
when its atoms are callable and target-compatible. Its coherent part must be
absorbed into the complete labels or decomposed into additional calibrated,
costed, target-compatible labels.
After this restriction, the dimension of the atom span modulo the compatible
coherent edge must be recomputed for the finite-support bound. Proposition~\ref{prop:compatible-attainment} in
Appendix~\ref{app:static-proofs} gives sufficient conditions for finite support.

Appendix~\ref{sec:direct} gives direct-label schedules with product-error,
interstitial, and dwell bounds for nominal synthesis. The next section also
covers atoms implemented through finite pulses.

\section{Robust compilation from primitive physical paths}
\label{sec:physical-words}

The static Lyapunov certificate supplies a decaying reference evolution.
To certify a physical word, we compare its actual cycle with a reference
word, then compare that reference word with the static evolution, using the
static metric for both comparisons. Bounds on partial words control the
intermediate motion that a cycle estimate alone does not describe.
We first treat target-preserving paths and then use shadow paths to bound
leakage.

Figure~\ref{fig:certificate-paths} organizes the two comparisons and shows
where the prefix estimate enters the all-time argument.

\begin{figure}[htbp]
\centering
\includegraphics[width=\linewidth]{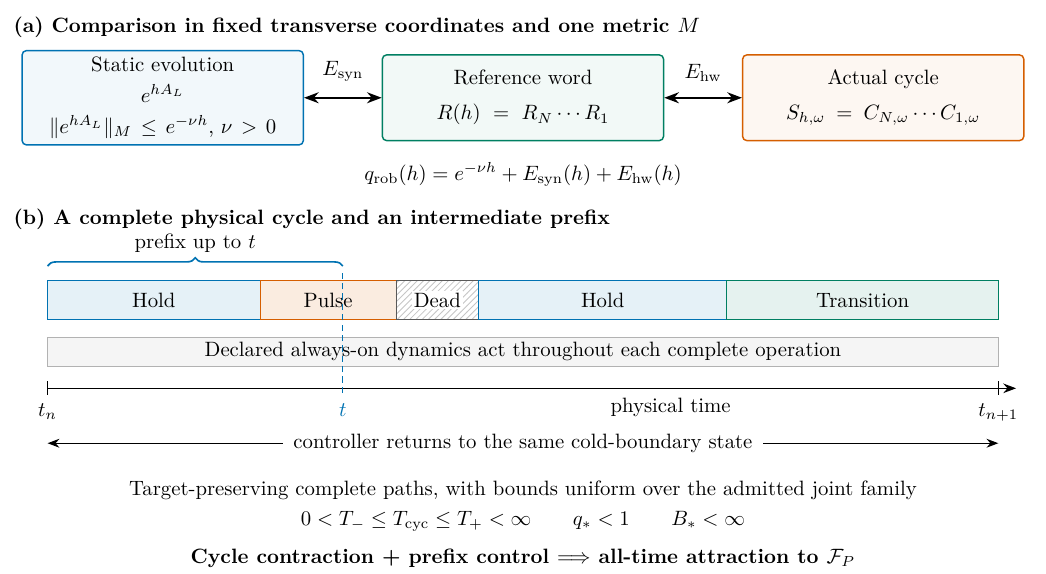}
\caption{Certificate comparisons and a complete physical path.
(a) The synthesis and hardware errors compare the three maps in the same
transverse coordinates and Lyapunov metric $M$
[Eqs.~\eqref{eq:derived-nominal-synthesis} and
\eqref{eq:robust-word-prefix-bounds}].
(b) An illustrative target-preserving cycle and a prefix ending during a
pulse. Operation widths are schematic. Concurrent background dynamics enter
each complete path. Controller closure and the uniform duration, contraction,
and prefix bounds are the hypotheses used in
Theorem~\ref{thm:robust-path-compiler} to obtain attraction at all physical
times. Here $q_*$ and $B_*$ are the uniform bounds in
Eq.~\eqref{eq:uniform-robust-slack}.}
\label{fig:certificate-paths}
\end{figure}

Bounded Eulerian control
replaces impulsive group jumps by continuous paths \cite{viola2003bounded}. In soft-pulse
Markovian models, the control Hamiltonian and dissipation act simultaneously
during the pulse \cite{pryadko2009soft}. The physical-word theorem therefore takes complete paths as input.

\subsection{Orbit access and physical paths}

\begin{definition}[Path-compatible word]
\label{def:path-compatible-word}
A physical word is path-compatible with $\F_P$ when every complete hold,
pulse, transition, dead interval, and cooldown path maps $\F_P$ into itself at
every intermediate time, every concatenation is valid in the declared
Markovian or reset model, and the final controller state equals the initial
cold state field by field.  The word channel contains every such path in
chronological order.
\end{definition}

For a CPTP map $\Phi(\rho)=\sum_aK_a\rho K_a^\dagger$,
\begin{equation}
 \Phi(\F_P)\subseteq\F_P
 \quad\Longleftrightarrow\quad QK_aP=0\quad\text{for every }a.
\label{eq:cptp-face-criterion}
\end{equation}
Along a complete path, the channel test must hold at every phase. For a
bounded piecewise-regular time-dependent GKLS path, invariance for every
target-face state injected at every phase is equivalent, almost everywhere,
to the time-dependent form of Eq.~\eqref{eq:gkls-face-criterion}.  A unitary
path therefore preserves the target face exactly when $QH(s)P=0$ almost
everywhere, or equivalently when
\begin{equation}
 U(s)\in\Stab(P)=\{U:UPU^\dagger=P\}
 \quad\text{for every }s.
\label{eq:unitary-stabilizer-path}
\end{equation}
Appendix~\ref{app:target-proofs} proves these channel and path criteria.

\begin{samepage}
\begin{theorem}[Orbit-to-word obstruction]
\label{thm:orbit-word-no-go}
The two cases below distinguish excitation that is switched off during pulses
from excitation that is unavoidably present during every bounded segment.
\begin{enumerate}
\item There is a qubit presentation containing an available excitation
generator and bounded transverse pulses, with excitation switched off during
the pulses, for which target-stabilizing amplitude damping is a unitary-orbit
atom and has a finite-cost stroboscopic realization, but no physical word of
positive total duration preserves the pure target at all times.
\item In the corresponding presentation with excitation unavoidably present
during every bounded segment and no noninvertible reset available, no finite
physical word of positive total physical duration returns exactly to the
target even stroboscopically.
\end{enumerate}
\end{theorem}
\end{samepage}

To exhibit the two cases of Theorem~\ref{thm:orbit-word-no-go}, define
\begin{equation}
 D_\uparrow=\Diss{|1\rangle\langle0|},\qquad
 D_\downarrow=\Diss{|0\rangle\langle1|},
\label{eq:up-down-damping}
\end{equation}
write $P_j=|j\rangle\langle j|$, and let
$\mathcal X(\rho)=X\rho X$.  Then
$\mathcal XD_\uparrow\mathcal X=D_\downarrow$.  If $D_\uparrow$ is off during
pulses generated by $H_\pm=\pm\Omega X/2$ with $\Omega>0$, two half turns
of duration $\pi/\Omega$ give
\begin{equation}
 \mathcal X e^{tD_\uparrow}\mathcal X=e^{tD_\downarrow},
 \qquad t>0.
\label{eq:qubit-stroboscopic-sandwich}
\end{equation}
During the first pulse, however,
\begin{equation}
 \dtr(U(s)P_0U(s)^\dagger,P_0)
 =\left|\sin\frac{\Omega s}{2}\right|,
\label{eq:qubit-pulse-excursion}
\end{equation}
which reaches one. The period $t+2\pi/\Omega$ and finite call count give a
finite-cost stroboscopic realization despite this loss of target compatibility.
If $\gamma D_\uparrow$, $\gamma>0$, is always on, the first nonzero segment
makes the state strictly positive, and finite-time GKLS propagators preserve
strict positivity. Appendix~\ref{subsec:orbit-word-proof} proves both statements.

\subsection{Uncertain primitive paths and block enclosures}

We now derive the primitive bounds for generator-driven presentations in
Definition~\ref{def:presentation}. Fix a family from
Definition~\ref{def:uncertain-word-family} in which every positive-duration
actual and reference path has a complete GKLS generator, and use the
transverse basis and metric $M$ from the static certificate. All estimates below must
hold throughout the admitted joint family. We suppress the exposure argument
in individual block data.

For the exact structure-preserving class, every complete actual and reference
path satisfies, almost everywhere,
\begin{equation}
 QV_a(t)P=0,\qquad
 Q\left(H(t)-\frac{i}{2}\sum_aV_a^\dagger(t)V_a(t)\right)P=0.
\label{eq:robust-path-face-conditions}
\end{equation}
Every actual and reference prefix then induces a transverse map.  For
primitive $k$, let $\chi_{k,\omega}:[0,1]\to[0,\tau_{k,\omega}]$ and
$\widehat\chi_k:[0,1]\to[0,\widehat\tau_k]$ be absolutely continuous,
nondecreasing clocks satisfying
\[
 \chi_{k,\omega}(0)=0,\qquad\chi_{k,\omega}(1)=\tau_{k,\omega},\qquad \widehat\chi_k(0)=0,\qquad\widehat\chi_k(1)=\widehat\tau_k.
\]
These clocks put the complete actual and reference paths, despite their
different physical durations, on the same normalized local time $u\in[0,1]$.
For positive-duration paths, let
$A^{\rm tr}_{k,\omega}$ and $\widehat A_k^{\rm tr}$ be the complete transverse
generators and set
\begin{equation}
 G_{k,\omega}(u)=\dot\chi_{k,\omega}(u)
 A^{\rm tr}_{k,\omega}(\chi_{k,\omega}(u)),\qquad \widehat G_k(u)=\dot{\widehat\chi}_k(u)
 \widehat A_k^{\rm tr}(\widehat\chi_k(u)).
\label{eq:normalized-generators}
\end{equation}
Thus amplitude, waveform, and duration error all enter the generator
difference on one interval. For a zero-duration identity slot, the clock and
normalized generator are identically zero, the normalized propagator is the
identity, and every controller field is unchanged, as in
Definition~\ref{def:uncertain-word-family}. The matrix logarithmic norm associated with
Eq.~\eqref{eq:M-induced-norm} is
\begin{equation}
 \mu_M(G)=\frac12\lambda_{\max}\!\left[
 M^{-1/2}(G^\dagger M+MG)M^{-1/2}\right].
\label{eq:M-logarithmic-norm}
\end{equation}

\begin{lemma}[Primitive-path enclosure]
\label{lem:primitive-duhamel}
Let $\Phi_{k,\omega}$ and $\widehat\Phi_k$ be the normalized propagators of
$G_{k,\omega}$ and $\widehat G_k$, with endpoints
$C_{k,\omega}=\Phi_{k,\omega}(1,0)$ and
$R_k=\widehat\Phi_k(1,0)$. Suppose the integrable envelopes
$\alpha_k,\widehat\alpha_k,e_k$ on $[0,1]$ satisfy, for almost every $u\in[0,1]$
and uniformly in $h\in H$ and $\omega\in\Omega(h)$,
\begin{equation}
 \mu_M(G_{k,\omega})\le\alpha_k,\qquad
 \mu_M(\widehat G_k)\le\widehat\alpha_k,\qquad
 \|G_{k,\omega}-\widehat G_k\|_M\le e_k.
\label{eq:primitive-envelope-data}
\end{equation}
Define
\begin{equation}
\begin{aligned}
 &c_k=\exp\!\int_0^1\alpha_k(v)\,dv,\qquad
 r_k=\exp\!\int_0^1\widehat\alpha_k(v)\,dv,\qquad
 p_k(u)=\exp\!\int_0^u\alpha_k(v)\,dv,\\
 &\delta_k^{\rm int}=\int_0^1
 \exp\!\left(\int_u^1\alpha_k(v)dv\right)e_k(u)\times
 \exp\!\left(\int_0^u\widehat\alpha_k(v)dv\right)du.
\end{aligned}
\label{eq:primitive-enclosures}
\end{equation}
Then with $\delta_k=\delta_k^{\rm int}$,
\begin{equation}
 \|C_{k,\omega}-R_k\|_M\le\delta_k,\qquad
 \|C_{k,\omega}\|_M\le c_k,\qquad \|R_k\|_M\le r_k,
\label{eq:primitive-block-bounds}
\end{equation}
and every actual partial block obeys
$\|\Phi_{k,\omega}(u,0)\|_M\le p_k(u)$ for every $u\in[0,1]$.
\end{lemma}

\begin{proof}
The error of an actual primitive relative to its reference path is controlled
by the variation-of-constants identity
\begin{equation}
 C_{k,\omega}-R_k=\int_0^1\Phi_{k,\omega}(1,u)
 (G_{k,\omega}(u)-\widehat G_k(u))
 \widehat\Phi_k(u,0)\,du.
\label{eq:primitive-duhamel-identity}
\end{equation}
Together with the logarithmic-norm propagator bounds, this proves the lemma.
\end{proof}

In particular, always-on dynamics during a pulse and dynamics during dead
time are included in the integrand rather than added after an endpoint
calculation. This treatment uses time-local chronological propagation for
finite-dimensional master equations \cite{chruscinski2010time_local}.

\subsection{From block errors to cycle and prefix bounds}

We first derive the nominal synthesis bound from finite product and block
bounds. Let
$W_1,\ldots,W_N$ be ideal splitting, commutator, or conjugation factors.  If a
stated finite product estimate gives
\begin{equation}
 \|W_N\cdots W_1-e^{hA_L}\|_M\le E_{\rm split}(h),
\label{eq:finite-splitting-receipt}
\end{equation}
and finite block calculations give
$\|R_k-W_k\|_M\le a_k$, $\|R_k\|_M\le r_k$, and
$\|W_k\|_M\le w_k$, define
\begin{equation}
 E_{\rm syn}(h)=E_{\rm split}(h)+
 \sum_{k=1}^N\left(\prod_{\ell>k}r_\ell\right)a_k
 \left(\prod_{\ell<k}w_\ell\right).
\label{eq:derived-nominal-synthesis}
\end{equation}
The product telescope proves
$\|R_N\cdots R_1-e^{hA_L}\|_M\le E_{\rm syn}(h)$.

We next propagate the actual--reference errors through the word. The
transverse map of a complete physical cycle is its \emph{actual monodromy}.

\begin{lemma}[Amplified word and prefix bounds]
\label{lem:amplified-word}
For $S_{h,\omega}=C_{N,\omega}\cdots C_{1,\omega}$ and $R(h)=R_N\cdots R_1$,
assume the endpoint and prefix inequalities of Lemma~\ref{lem:primitive-duhamel}
hold with certified bounds $c_k,r_k,p_k,\delta_k$. Set
\begin{equation}
 E_{\rm hw}(h)=\sum_{k=1}^N
 \left(\prod_{\ell>k}c_\ell\right)\delta_k
 \left(\prod_{\ell<k}r_\ell\right),\qquad B_{\rm rob}(h)=\max\left\{1,
 \sup_{k,u}p_k(u)\prod_{\ell<k}c_\ell\right\}.
\label{eq:robust-word-prefix-bounds}
\end{equation}
Then
$\|S_{h,\omega}-R(h)\|_M\le E_{\rm hw}(h)$, and every actual
within-cycle prefix has induced $M$-norm at most $B_{\rm rob}(h)$.
\end{lemma}

The exact identity
\begin{equation}
 C_N\cdots C_1-R_N\cdots R_1
 ={}\sum_{k=1}^NC_N\cdots C_{k+1}
 {}\times(C_k-R_k)R_{k-1}\cdots R_1
\label{eq:robust-word-telescope}
\end{equation}
proves the first assertion: each primitive error is multiplied by its actual
suffix and reference prefix, and both corresponding amplification factors
are needed for nonnormal words. Direct multiplication of completed blocks
and one partial block gives the second assertion.

\subsection{Exact target-preserving compilation}

Choose finite constants $K_{\rm in},K_{\rm out}>0$ such that
\begin{equation}
 \|z(\rho)\|_M\le K_{\rm in}\dtr(\rho,\F_P),\qquad
 \dtr(\rho,\F_P)\le K_{\rm out}\|z(\rho)\|_M.
\label{eq:norm-conversions}
\end{equation}
These constants convert transverse decay into target distance for the fixed
quotient basis and metric. Using that metric for the static, primitive, and
cycle estimates lets the uniform contraction bound iterate under cycle-varying
uncertainty.

\begin{theorem}[Primitive-certified compiler]
\label{thm:robust-path-compiler}
Let an uncertain word family as in Definition~\ref{def:uncertain-word-family}
satisfy Eq.~\eqref{eq:robust-path-face-conditions} and robust cold closure.  Suppose the
same metric $M$ satisfies
\begin{equation}
 A_L^\dagger M+MA_L\preceq-2\nu M,\qquad \nu>0,
\label{eq:compiler-static-lyapunov}
\end{equation}
with integrable envelopes in Eq.~\eqref{eq:primitive-envelope-data} and
finite, uniform bounds in Eqs.~\eqref{eq:derived-nominal-synthesis} and
\eqref{eq:robust-word-prefix-bounds}. Define
\begin{equation}
 q_{\rm rob}(h)=e^{-\nu h}+E_{\rm syn}(h)+E_{\rm hw}(h).
\label{eq:robust-contraction-factor}
\end{equation}
\begin{enumerate}
\item \emph{Single-cycle bound.}
For every $h\in H$ and $\omega\in\Omega(h)$, the actual monodromy satisfies
\begin{equation}
 S_{h,\omega}^\dagger MS_{h,\omega}
 \preceq q_{\rm rob}(h)^2M.
\label{eq:derived-actual-monodromy}
\end{equation}
\item \emph{All-time attraction.}
If an admitted infinite cycle sequence is executable and
\begin{equation}
 q_*:=\sup_{h\in H}q_{\rm rob}(h)\le1-\gamma<1,\qquad
 B_*:=\sup_{h\in H}B_{\rm rob}(h)<\infty,
\label{eq:uniform-robust-slack}
\end{equation}
then for every executable admitted sequence $(h_n,\omega_n)$ and initial
density operator $\rho(0)$, the trajectory satisfies, for $t\in[t_n,t_{n+1})$
with $t_n$ the start of cycle $n$,
\begin{equation}
 \dtr(\rho(t),\F_P)\le
 K_{\rm out}B_*q_*^nK_{\rm in}\dtr(\rho(0),\F_P).
\label{eq:robust-all-time-bound}
\end{equation}
Along each such sequence, $\F_P$ is all-time invariant and GAS. For $r=1$,
its unique pure state is all-time GAS. For $r>1$, attraction is to the face.
\end{enumerate}
\end{theorem}

\begin{proof}
Equation~\eqref{eq:compiler-static-lyapunov} gives $\|e^{hA_L}\|_M\le e^{-\nu h}$.
The triangle inequality, nominal bound, and Lemma~\ref{lem:amplified-word} give
$\|S_{h,\omega}\|_M\le q_{\rm rob}(h)$, equivalent to
Eq.~\eqref{eq:derived-actual-monodromy}. Iterating the cycle-varying maps gives
$\|z_n\|_M\le q_*^n\|z_0\|_M$ for $z_n=z(\rho(t_n))$. The prefix bound and
Eq.~\eqref{eq:norm-conversions} give Eq.~\eqref{eq:robust-all-time-bound}.
Path compatibility gives invariance. Equation~\eqref{eq:nonzeno-cycle-duration}
gives $nT_-\le t_n\le nT_+$. The lower bound excludes finite accumulation of
cycle boundaries, and $n\ge t/T_+-1$ for $t\in[t_n,t_{n+1})$ ensures
convergence in physical time.
\end{proof}

For a higher-rank target, Appendix~\ref{app:monodromy} gives an intraface
contraction certificate for selecting a unique mixed state at cycle boundaries. If all admitted prefixes
are CPTP and fix that state, convergence holds throughout each cycle.

\subsection{Bounded leakage through shadow-path comparison}

Unrestricted GKLS perturbations need not preserve exact invariance, even at
arbitrarily small positive jump rates, as the qubit counterexample in
Eq.~\eqref{eq:leakage-jump-example} shows. We bound the resulting practical
tube by comparison with target-preserving dynamics.

For each leaky actual primitive, choose a complete structure-preserving
\emph{shadow path} with the same primitive type and normalized clock.  Suppose
their complete normalized GKLS generators obey
\begin{equation}
 \|\mathcal G_{k,\omega}(u)-
 \widetilde{\mathcal G}_{k,\omega}(u)\|_\diamond
 \le e_k^\diamond(u).
\label{eq:primitive-diamond-envelope}
\end{equation}
Write $(\mathcal C_{k,h,\omega},\widetilde{\mathcal C}_{k,h,\omega})$,
$(\mathcal S_{h,\omega},\widetilde{\mathcal S}_{h,\omega})$, and
$(\mathcal P_{h,\omega}^{k,u},\widetilde{\mathcal P}_{h,\omega}^{k,u})$ for the
actual and shadow block endpoints, cycles, and prefixes through local time $u$
of block $k$, respectively. Take full-diamond suprema over the admitted paths,
using $0$ as the supremum of the empty family in $[0,\infty)$:
\begin{equation}
 \epsilon_k^\diamond=\sup_{h,\omega}
 \|\mathcal C_{k,h,\omega}-\widetilde{\mathcal C}_{k,h,\omega}\|_\diamond,\qquad \epsilon_{\rm cyc}^\diamond=\sup_{h,\omega}
 \|\mathcal S_{h,\omega}-\widetilde{\mathcal S}_{h,\omega}\|_\diamond,\qquad \epsilon_{\rm pre}^\diamond=\sup_{h,\omega,k,u}
 \|\mathcal P_{h,\omega}^{k,u}-
 \widetilde{\mathcal P}_{h,\omega}^{k,u}\|_\diamond.
\label{eq:diamond-distance-definitions}
\end{equation}
For a nonempty paired family with nonnegative integrable envelopes
$e_k^\diamond$ on $[0,1]$, Duhamel's identity, a channel telescope, and
the unit diamond norm of CPTP propagators give
\begin{equation}
 \epsilon_k^\diamond\le\int_0^1e_k^\diamond(u)\,du,\qquad \epsilon_{\rm cyc}^\diamond\le\sum_{k=1}^N\epsilon_k^\diamond,\qquad \epsilon_{\rm pre}^\diamond\le\sup_{k,u}\left(
 \sum_{\ell<k}\epsilon_\ell^\diamond+
 \int_0^u e_k^\diamond(v)\,dv\right).
\label{eq:diamond-cycle-prefix-leakage}
\end{equation}
Define the norm-conversion constant
\begin{equation}
 \kappa_{1M}=\sup_{\substack{X=X^\dagger\\\|X\|_1\le1}}
 \|z(X)\|_M,\qquad
 \eta_{\rm cyc}^M=\kappa_{1M}\epsilon_{\rm cyc}^\diamond,\qquad
 \eta_{\rm pre}^M=\kappa_{1M}\epsilon_{\rm pre}^\diamond.
\label{eq:leakage-M-conversion}
\end{equation}

\begin{theorem}[Primitive-certified practical tube]
\label{thm:primitive-practical-tube}
Let the actual words satisfy Definition~\ref{def:uncertain-word-family} and the
shadow words satisfy Theorem~\ref{thm:robust-path-compiler} with $q_*<1$ and $B_*$.
Assume the pairs satisfy Eqs.~\eqref{eq:primitive-diamond-envelope}--\eqref{eq:leakage-M-conversion}
uniformly, including under cycle-varying uncertainty. Put
$d_0=\dtr(\rho(0),\F_P)$.  For every $t\in[t_n,t_{n+1})$,
\begin{equation}
 \dtr(\rho(t),\F_P)\le K_{\rm out}\Bigg[
 B_*\Bigg(q_*^nK_{\rm in}d_0
 {}+\eta_{\rm cyc}^M\frac{1-q_*^n}{1-q_*}\Bigg)
 +\eta_{\rm pre}^M\Bigg].
\label{eq:primitive-practical-tube}
\end{equation}
The corresponding asymptotic upper radius is
\begin{equation}
 R_\infty=K_{\rm out}\left[
 \frac{B_*\eta_{\rm cyc}^M}{1-q_*}+\eta_{\rm pre}^M\right].
\label{eq:practical-radius}
\end{equation}
\end{theorem}

The proof compares the actual and shadow cycle channels, and their
within-cycle prefixes, on the same actual state at the beginning of each
cycle. Comparing states in this way remains valid when the leaky path does
not induce a map on the transverse quotient.
For the cycle-boundary quantities $y_n=\|z(\rho(t_n))\|_M$, the comparison
gives $y_{n+1}\le q_*y_n+\eta_{\rm cyc}^M$. At an intermediate phase it gives
$\|z(\rho(t))\|_M\le B_*y_n+\eta_{\rm pre}^M$.  Summing the geometric series
proves the theorem.

When $\eta_{\rm cyc}^M=0$ but $\eta_{\rm pre}^M>0$,
the asymptotic intracycle error is at most $K_{\rm out}\eta_{\rm pre}^M$.
A nonzero asymptotic excursion from $\F_P$ requires a separate positive lower
bound, while cycle-boundary convergence is certified.

\subsection{Timing and resource guarantees}

When the primitive constraints for a word family parameterized by $h$
supply lower exposure endpoints, set
\begin{equation}
 h_{\min}=\max\{h_{\min}^{\rm dwell},h_{\min}^{\rm dead},
 h_{\min}^{\rm pulse},h_{\min}^{\rm cooldown},
 h_{\min}^{\rm calibration}\}.
\label{eq:physical-hmin}
\end{equation}
Use a nonempty calibrated domain of realizable exposures satisfying the
primitive constraints and Eq.~\eqref{eq:uniform-robust-slack}.
A fixed positive hardware error or overhead can make $h\downarrow0$ infeasible:
if
$\liminf_{h\downarrow0}E_{\rm hw}(h)\ge\delta_0>0$, then for all sufficiently
small $h$, $q_{\rm rob}(h)\ge e^{-\nu h}+\delta_0/2>1$.

For the resource account, set time zero at the first certified cold boundary
after initialization.  Let $K_{\rm eff}=K_{\rm out}B_*K_{\rm in}$ and let
$d_0$ be the target distance at that boundary or a valid upper bound on it.
For $0<q_*<1$ in the exact branch and $K_{\rm eff}d_0>\eps>0$,
the certified cycle count is
$\lceil\log(K_{\rm eff}d_0/\eps)/(-\log q_*)\rceil$.
For the leaky branch, put
$D_0=K_{\rm out}(B_*K_{\rm in}d_0+\eta_{\rm pre}^M)$.
When $0<q_*<1$ and $D_0>\eps>R_\infty$, the count is
$\lceil\log[(D_0-R_\infty)/(\eps-R_\infty)]/(-\log q_*)\rceil$.
Appendix~\ref{app:monodromy} records the complete branches, including
zero-cycle, zero-contraction, and threshold cases.

Let $T_*$ and $C_{\rm cyc}^*$ be finite certified upper bounds on cycle
duration and componentwise cost over the joint uncertainty family. Continuous
coordinates attain their worst-case values on a nonempty compact family.
Initialization has its own time and resource costs. Any state-changing
finalization requires a separate certificate, with its costs included when
finalization is part of the specified task.
For either branch, let $N_\eps$ denote the corresponding finite certified
cycle count specified in Appendix~\ref{app:monodromy}. Then
\begin{equation}
 T_{\rm wall}\le T_{\rm init}+N_\eps T_*+T_{\rm final},\qquad
 C(\eps)\preceq C_{\rm init}+N_\eps C_{\rm cyc}^*+C_{\rm final}.
\label{eq:componentwise-cost}
\end{equation}
The integer ceiling is evaluated with the worst-case $q_*$ and leakage radius.
Logical-orbit, raw, pulse, transition, reset, dead,
cooldown, and switch counts are componentwise accounts and may overlap. An
aggregate primitive count records each physical call once.  The integrated
control action is an amplitude account in a fixed normalization, distinct from
thermodynamic work.  With
$H_c(t)=\sum_\alpha u_\alpha(t)H_\alpha$ in the declared Hamiltonian basis, it uses
$\|u(t)\|_{\rm ctrl}=\sum_\alpha|u_\alpha(t)|$ in
Eq.~\eqref{eq:integrated-control-action}.

\section{A common-ground qutrit realization}
\label{sec:qutrits}
\label{subsec:common-ground}

This example supplies the complete paths, reference contraction, gain-error
bounds, and prefix bounds required by
Theorem~\ref{thm:robust-path-compiler}. The two decay directions share a ground
state, allowing the conjugating pulses to preserve the target throughout
execution.

\subsection{Physical realization and uncertainty}

In the ordered basis $\{|0\rangle,|1\rangle,|2\rangle\}$, set
\begin{equation}
 A=\Diss{|0\rangle\langle1|},\qquad
 C=\Diss{|0\rangle\langle2|},\qquad
 P_j=|j\rangle\langle j|\quad (0\le j\le2).
\label{eq:common-ground-generators}
\end{equation}
These generators commute and each fixes $P_0$. Their sum $A+C$ is GAS to
$P_0$, whereas either generator alone leaves an undamped excited mode. If both
are direct labels, the two-decay cycle equals the averaged evolution by
commutativity. The fixed-metric compiler instance considered here has access
only to the raw decay $A$ and uses bounded pulses to supply the second decay
direction through conjugation. Let
\begin{equation}
 S_{12}=P_0+|1\rangle\langle2|+|2\rangle\langle1|.
\label{eq:excited-swap}
\end{equation}
We choose a bounded unitary path whose endpoint channel
$\mathcal S_{12}(\rho)=S_{12}\rho S_{12}^\dagger$ exchanges the two excited
levels and whose entire path preserves $\spann\{|0\rangle\}$, so the path
fixes the density operator $P_0$ and lies in $\Stab(P_0)$ at every time.

The closed chronological word is
\begin{equation}
 A\;\longrightarrow\;S_{12}\;\longrightarrow\;A
 \;\longrightarrow\;S_{12}^{-1}.
\label{eq:common-ground-word}
\end{equation}

The word starts and ends at the same cold boundary, in the laboratory frame
with an idle actuator. Both directed swap commands vanish
at their endpoints, and the time-local Markovian environment requires no reset.
Appendix~\ref{app:qutrit-calculations} specifies the controller signatures and
complete triangular waveforms.

We use the operating point $h=1$, with each of the four segments lasting one
unit. The damping envelope $a(s)$ has exposure $\alpha=2\log(5/4)$.
The pulse Hamiltonians $H_+$ and $H_-$ have action $\pi/2$, peak amplitude
$\pi$, and finite slew $2\pi$.

To state the robust family, let $\mathcal Z=\Diss{P_0}$,
$\lambda=2\log(4/3)$, and choose $\lambda_{\rm ao}\in\{0,\lambda\}$. The
controller frames label the path endpoints. In the common laboratory basis,
the complete generators for gain multipliers $m_1,m_3$ are
\begin{equation}
\begin{aligned}
 &\mathcal L_1(s)=m_1a(s)A+\lambda_{\rm ao}\mathcal Z,\qquad
 \mathcal L_2(s)=-i[H_+(s),\,\cdot\,]+\lambda_{\rm ao}\mathcal Z,\\
 &\mathcal L_3(s)=m_3a(s)A+\lambda_{\rm ao}\mathcal Z,\qquad
 \mathcal L_4(s)=-i[H_-(s),\,\cdot\,]+\lambda_{\rm ao}\mathcal Z.
\end{aligned}
\label{eq:common-ground-complete-paths}
\end{equation}
The always-on term therefore acts during both holds and both pulses. The
uncertain gains range over the declared set
\begin{equation}
 (m_1,m_3)\in[19/20,21/20]^2.
\label{eq:common-ground-gain-box}
\end{equation}
This set has relative $\ell^\infty$ radius $\rho=1/20$ in the two gain
coordinates, and each cycle may use any gain pair in this box. Pulse and
timing coordinates are fixed in this witness.

\subsection{Uniform convergence and resources}

Use the orthonormal transverse basis of
Proposition~\ref{prop:common-ground-word} in Appendix~\ref{app:qutrit-calculations}.
The reference monodromy is the transverse map of one complete reference
cycle and, in the metric $M=I_8$, has norm $4/5$ without dephasing and $16/25$
with the stated always-on dephasing. Exact hold endpoint bounds transfer its
contraction bound to the full gain box, while prefix bounds control every
intermediate phase of execution.
\begin{proposition}[Common-ground robust radius]
\label{prop:common-ground-robust}
For every pair in Eq.~\eqref{eq:common-ground-gain-box}, chosen arbitrarily
in every cycle, the primitive-certified bounds at $h=1$ and $M=I_8$, with the
static Lyapunov triples in Appendix~\ref{app:qutrit-calculations}, satisfy
\begin{equation}
 E_{\rm syn}=0,\qquad \delta_1,\delta_3<\frac3{200},\qquad
 E_{\rm hw}<\frac3{100},\qquad B_{\rm rob}=1.
\label{eq:common-ground-robust-receipts}
\end{equation}
Let $S_n^{\rm act}$ denote the actual monodromy selected in cycle $n$.
Consequently
\begin{equation}
 \bar q=
 \begin{cases}
 83/100,&\lambda_{\rm ao}=0,\\
 67/100,&\lambda_{\rm ao}=\lambda,
 \end{cases}\qquad
 (S_n^{\rm act})^{\dagger}S_n^{\rm act}\preceq\bar q^2I_8.
\label{eq:common-ground-robust-contractions}
\end{equation}
The certified contraction margins $1-\bar q$ are $17/100$ and $33/100$. With
$K_{\rm in}=2$ and $K_{\rm out}=3/2$, every physical time in cycle $n$ obeys
\begin{equation}
 \dtr(\rho(t),P_0)\le3\,\bar q^{\,n}\dtr(\rho(0),P_0).
\label{eq:common-ground-robust-all-time}
\end{equation}
\end{proposition}

Appendix~\ref{app:qutrit-calculations} proves
Eq.~\eqref{eq:common-ground-robust-receipts} uniformly on the gain box by
deriving the static Lyapunov triples from the diagonal reference maps and
using the exact endpoint errors of the uncertain holds in
Lemma~\ref{lem:amplified-word}. Theorem~\ref{thm:robust-path-compiler} then
uses the cycle and prefix bounds to give the all-time estimate.

For $d_0=1$ and $\eps=10^{-3}$, the two contraction bounds certify $43$ and
$20$ cycles, respectively. Each cycle lasts four units and uses four
primitive calls, so the wall-time bounds are $172$ and $80$. The corresponding
integrated control actions are $43\pi$ and $20\pi$.
Table~\ref{tab:common-ground-costs} gives the full resource account under the
declared calibration and cost conventions.

Appendix~\ref{app:qutrit-calculations} also distinguishes a failed comparison
bound from actual contraction and shows why replacing a directly callable
decay with its finite-pulse implementation requires a new prefix calculation. The separate two-label cascade in
Appendix~\ref{subsec:ladder} admits exact finite factorization and matching
presentation-relative wall-time bounds with leading term $4\log(1/\eps)$.

\section{Searching a finite hardware library}
\label{sec:certified-search}

Complete-path certificates can also select executable words from a finite
hardware library. Switched dissipative stabilization and robust open-system
performance are studied in \cite{scaramuzza2015switching,schirmer2022robust}. We seek guarantees
over a specified library's full uncertainty set, including finite pulses and
idle dynamics.

\subsection{Acceptance and refinement in a finite library}

Fix a finite operation alphabet and a maximal word length.  For each word,
use a compact parameter domain whose coordinates enforce all exact target
and controller identities.  The uncertainty set covers every admitted joint
error, including correlations between operations.  The system model includes
every quantum degree of freedom retained between calls.  Each nonempty cycle
has duration between fixed positive finite bounds.

The norm certificates of Section~\ref{sec:physical-words} supply the two
progress bounds used below: contraction controls repeated cycle boundaries,
while a prefix bound extends the guarantee through each operation. The test
also accepts a positive-observable certificate, as used in the four-level
example.

\begin{proposition}[A finite acceptance test]
\label{prop:finite-search-test}
Suppose a complete-path calculation gives the following bounds for a
nonnegative state functional $V$.
\begin{equation}
 V(\mathcal S\rho)\le qV(\rho)+\eta,\qquad
 V(\mathcal P\rho)\le B V(\rho)+b,\qquad
 d_{\rm tr}(\rho,\F_P)\le\psi(V(\rho)).
\label{eq:search-progress-data}
\end{equation}
In every admitted repeated execution, the first two bounds hold at every
reachable cycle input $\rho$, with $\mathcal S$ the channel of that cycle and
$\mathcal P$ any of its physical prefixes. The third bound holds at every
corresponding prefix output, including the cycle boundaries. Here
$0\le q<1$, $B,\eta,b\ge0$, and $\psi$ is increasing on a domain containing
the arguments in the distance bound and the test below.
Let $V(\rho_0)\le Y$ and, with $q^0=1$ also when $q=0$, put
\begin{equation}
 z_\infty=\frac{\eta}{1-q},\qquad
 z_N=q^N Y+(1-q^N)z_\infty.
\label{eq:search-tail-sequence}
\end{equation}
An integer $N\ge0$ satisfying
\begin{equation}
 \psi\!\left(B\max\{z_N,z_\infty\}+b\right)\le\eps
\label{eq:search-acceptance}
\end{equation}
certifies the tolerance from the boundary after $N$ completed cycles onward
under continued execution. If every admitted cycle takes at most $T$, this
boundary is reached by time $NT$ after initialization.
\end{proposition}

\begin{proof}
Iteration of the first inequality gives $V(\rho_n)\le z_n$.  For $0\le q<1$,
the sequence $z_n$ lies monotonically between $Y$ and $z_\infty$, so its
supremum for $n\ge N$ is $\max\{z_N,z_\infty\}$.  Apply the prefix bound
and the monotonicity of $\psi$.  Summing the bounded cycle durations proves
the time statement.
\end{proof}

For a pure target, the bound below reduces
Eq.~\eqref{eq:search-acceptance} to a scalar test for a survival certificate,
with exact rational arithmetic for rational bounds and tolerance.
Taking $V(\rho)=\tr(Q\rho)$ gives
\begin{equation}
 V(\rho)\le d_{\rm tr}(\rho,\F_P)\le\sqrt{V(\rho)}.
\label{eq:survival-distance}
\end{equation}
For the upper inequality, the case $V=1$ follows from the universal bound
one. For $V<1$, purify $\rho$ and project its purification onto $P$ on the
system factor. The normalized projected vector has overlap
$\sqrt{1-V(\rho)}$ with the original vector and trace distance
$\sqrt{V(\rho)}$, which partial trace cannot increase. The lower inequality
follows from measuring $\{P,Q\}$ for every state.

For continuous parameters, fair interval refinement explores both parameter
boxes and arithmetic accuracy. Under the compactness and enclosure-convergence
assumptions of Proposition~\ref{prop:search-refinement} in
Appendix~\ref{app:four-level-certification}, it finds a certificate whenever a
represented strictly feasible point exists, while boundary certificates may
remain unresolved. Once a certified execution-time upper bound is available,
the positive lower cycle-duration bound limits the integer counts that can
improve it.
\begin{remark}[Finite-instance scope]
\label{rem:finite-instance-scope}
\label{rem:robust-compiler-scope}
With fixed finite support, graph, and word order, a prescribed maximum word
length, candidate paths drawn from a declared closed-form library, and algebraic
or rational enclosures, the compiler certifies supplied data by finite symbolic
calculations of equalities, positivity conditions, timing inequalities, and
integer cycle counts. This scope excludes generic matrix-exponential identities
and global minimum cost. Real quantifier elimination decides a separately supplied
finite first-order formula over the real algebraic numbers. These results
supply no running-time bound or witness-size classification.

Failure to find a certificate, or failure of a sufficient contraction bound,
leaves actual stability and physical feasibility unresolved unless a necessary
analytic obstruction applies.

\end{remark}

\subsection{Four-level model and admitted errors}
\label{subsec:four-level-search}

For orthonormal levels $|0\rangle,\ldots,|3\rangle$, let $P_j=|j\rangle\langle j|$
and choose target $P_0$. Define
\begin{equation}
 \mathcal A=\Diss{|0\rangle\langle1|}+\Diss{|1\rangle\langle2|},\quad
 H_{23}=|2\rangle\langle3|+|3\rangle\langle2|,\quad
 \mathcal Z=\sum_{j=0}^3\Diss{P_j}.
\label{eq:four-level-hardware}
\end{equation}
Let $g$ be the dimensionless cascade gain, $a$ the dimensionless pulse gain,
and $\kappa$ the dephasing rate. In the shadow model, a hold $A(t)$ applies
$g\mathcal A+\kappa\mathcal Z$.  The pulses $P(\theta)$ and $M(\theta)$ apply
$-i[+4aH_{23},\cdot]+\kappa\mathcal Z$ and
$-i[-4aH_{23},\cdot]+\kappa\mathcal Z$, respectively, for their uncertain
finite durations.
Each actual command duration is an uncertain multiple of its nominal
value, $t$ for a hold and $\theta/4$ for a pulse.
Every command is followed by a fixed dead interval with the same dephasing.
Commands obey $t\in[0.3,6]$ and $\theta\in[0.6,2]$.
Table~\ref{tab:four-level-guarantee} specifies the numerical error ranges
and background rates.

These dimensionless parameters define a theoretical hardware instance, with
time measured in inverse nominal cascade rate. The amplitude bounds are
finite, and the control model is
piecewise constant with no additional slew-rate constraint.  The cascade is
off during pulses, while dephasing is always on.  All gains and durations are
constant within a segment and may have arbitrary correlations across
segments and cycles within the Cartesian box of Table~\ref{tab:four-level-guarantee}.

The actual model adds a further jump
$\lambda(s)\Diss{|1\rangle\langle0|}$ to the shadow dynamics throughout every
path, with the bound in Table~\ref{tab:four-level-guarantee}.
There are no auxiliary degrees of freedom or resets.  Commands are issued in
a fixed laboratory basis and return the actuator to idle after each call.
The sequence counter closes after a complete word.  No exact inverse of an
inaccurate quantum pulse is required for this classical controller closure.

\subsection{Survival certificates and selected words}

The cascade and finite pulse do not commute. A hold retains the dark state
$P_3$, whereas a pulse retains $P_1$. To bound the probability remaining
outside the target after a complete cycle, use the survival effect
$E=\mathcal S^*(Q)$ of the shadow cycle channel $\mathcal S$. This Hermitian
operator satisfies $0\preceq E\preceq I$ and
$\tr(E\rho)=\tr(Q\mathcal S\rho)$. Starting from $Q$, the shadow adjoint
dynamics closes on a four-dimensional real space of Hermitian operators that
retains the $2$--$3$ coherence created by the pulse.
Appendix~\ref{app:four-level-certification} derives this reduction and its
uniform interval bounds, including how intersecting valid enclosures from
two algebraically equivalent pulse formulas retains useful cancellation near
a swap.

\begin{samepage}
The simple word $A(3),P(157/100),A(5)$ gives a two-cycle preparation-time
bound of $17.4207$. To select a word from the declared finite library, the
same rational checker evaluates the four templates $APA$, $APAP$, $APAM$,
and $APAPA$ on the hold grid $\{2,3,4,5,6\}$ and pulse-angle grid
$\{1.5,1.57,1.65\}$, giving 1650 parameterized words. The smallest certified
time in that library belongs to
\begin{equation}
 A(3),\ P(157/100),\ A(6),\ P(157/100),\ A(4).
\label{eq:four-level-five-word}
\end{equation}
Table~\ref{tab:four-level-guarantee} gives its certificate and all-time
guarantee over the whole declared error box.
\par
\end{samepage}

Figure~\ref{fig:four-level-protocol} connects the hardware and selected word
to the survival certificate, including the bound between cycle boundaries.

\begin{figure}[tbp]
\centering
\includegraphics[width=\linewidth]{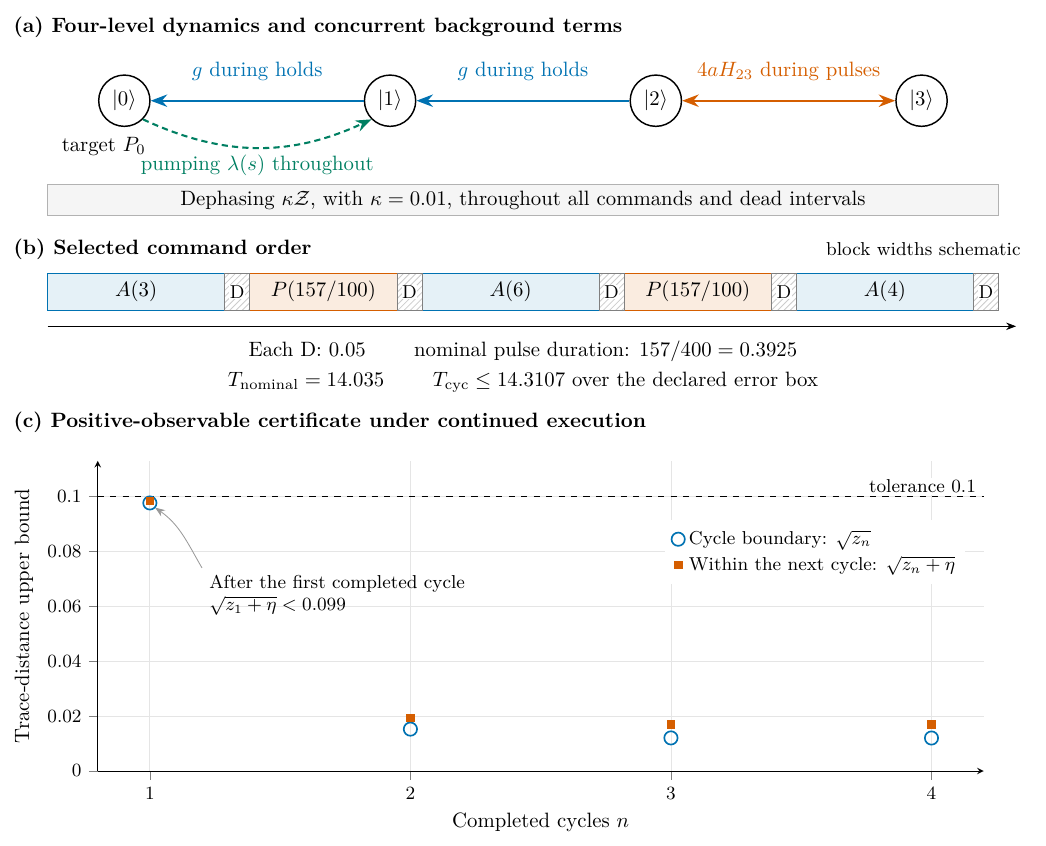}
\caption{Four-level execution and its positive-observable certificate under
the full error box in Table~\ref{tab:four-level-guarantee}.
(a) The cascade, pulse coupling, pumping, and dephasing in the actual model.
(b) The command order in Eq.~\eqref{eq:four-level-five-word}, including every
dead interval. Times are in inverse nominal cascade-rate units.
(c) Certified trace-distance upper bounds for every initial density operator
under continued execution. With $q=47/5000$, $T=14.3107$, and
$\eta=10^{-5}T$, Eq.~\eqref{eq:search-tail-sequence} gives
$z_n=q^n+(1-q^n)\eta/(1-q)$.
The markers show $\sqrt{z_n}$ at the boundary after $n$ completed cycles and
$\sqrt{z_n+\eta}$ throughout the following cycle, using
Eq.~\eqref{eq:survival-distance} and the prefix bound in
Appendix~\ref{app:four-level-certification}.
The tolerance is certified from the first completed cycle onward, whose
boundary is reached by $T$.}
\label{fig:four-level-protocol}
\end{figure}

\begin{table}[tbp]
\caption{Certified operation of the five-command four-level word in
Eq.~\eqref{eq:four-level-five-word}. Times are in inverse nominal cascade-rate
units. The model is a theoretical hardware instance.}
\label{tab:four-level-guarantee}
\small
\begin{tabular}{p{0.25\linewidth}p{0.65\linewidth}}
\toprule
Quantity & Admitted values or conclusion\\
\midrule
Command errors & Cascade gain $g\in[0.95,1.05]$, pulse gain
$a\in[0.98,1.02]$, and command-duration multipliers in $[0.98,1.02]$.
The cascade gain is shared by its two jumps within each hold.\\[1ex]
Background and timing & Dephasing $\kappa=0.01$ throughout, pumping
$0\le\lambda(s)\le10^{-5}$ throughout, and a fixed dead interval $0.05$
after each command. The cascade is off during pulses.\\[1ex]
Shadow certificate & The adjoint survival bound has factor below $0.00926$.
Appendix~\ref{app:four-level-certification} gives a simpler outward bound
$47/5000$ sufficient for the same preparation guarantee.\\[1ex]
Actual guarantee & For every initial density operator, a complete
preparation cycle ends by $T=14.3107$ with $d_{\rm tr}(\rho(t),P_0)\le0.1$
at that time and throughout subsequent cycles.\\
\bottomrule
\end{tabular}
\end{table}

\noindent
\begin{minipage}{\linewidth}
\subsection{Comparison and a physical lower bound}

For comparison on the same hardware and uncertainty set, a fixed-periodic
comparator repeats $AP$ with one common hold and angle. It uses the same dead
intervals and primitive interval formulas, with the block-partition refinement
detailed in Appendix~\ref{app:periodic-continuation}.
On hold grid $\{1.50,1.55,\ldots,6.00\}$ and angle grid
$\{1.45,1.47,\ldots,1.69\}$, both scheduler phases and every preparation
prefix of one through twelve commands are checked. The best certified
prefix is $APAPAPA$, with hold $3.4$ and angle $1.55$.
Its all-time continuation certificate starts at the prefix's actual scheduler
phase and uses the acceptance inequality derived in
Appendix~\ref{app:periodic-continuation}. Relative to this finite-grid comparator,
the five-command word improves the reported guarantee by about $7.1\%$.

Applying the same block-partition refinement to all 1650 library words
improves 24 survival bounds but changes no accepted cycle count or
preparation time.
\end{minipage}

\begin{samepage}
Numerical optimization of the same certificate objective proposes a
five-command word, whose smaller time bound is certified by the same rational
checker (Table~\ref{tab:four-level-comparison}).
Appendix~\ref{app:optimized-word} gives its parameters and acceptance bound.
At tolerance $0.1$, these
certified upper bounds remain well above the hardware lower bound derived
below in Eq.~\eqref{eq:four-level-physical-lower}.
\par
\end{samepage}

\begin{table}[tbp]
\caption{Certified preparation-time upper bounds for the same four-level
hardware, uncertainty set, and tolerance $0.1$. Operation continues after
preparation. See Supplemental Material, Sec.~2, for the search domains.}
\label{tab:four-level-comparison}
\small
\begin{tabular}{p{0.32\linewidth}p{0.43\linewidth}r@{.}l}
\toprule
Search & Selected preparation and continuation & \multicolumn{2}{c}{Time bound}\\
\midrule
1650-word finite library & $APAPA$, repeated after preparation & $14$&$3107$\\[1ex]
Declared periodic $AP$ grid & $APAPAPA$ preparation, then matching-phase
$PAPA$ blocks thereafter & $15$&$40775$\\[1ex]
Numerical proposal search & Five-command $APAPA$, repeated after preparation & $14$&$024998$\\
\bottomrule
\end{tabular}
\end{table}

Finally, for every available actual path,
$\dot p\ge-(21/20)p$ with $p=\tr(Q\rho)$, since pumping adds a nonnegative
term and only the first cascade jump removes population from $Q$.
For $p(0)=1$, the bound $d_{\rm tr}(\rho(t),P_0)\ge p(t)\ge e^{-21t/20}$
rules out exact preparation in finite time. Reaching tolerance $0<\eps<1$
from this initial population requires
\begin{equation}
 T_{\rm physical}\ge\frac{20}{21}\log\frac1\eps.
\label{eq:four-level-physical-lower}
\end{equation}

\section{Conclusion}
\label{sec:conclusion}

The primitive-certified compiler combines a static Lyapunov estimate,
nominal synthesis, and complete uncertain path bounds. With target-preserving
paths, reproducible controller closure, uniformly strict contraction, bounded
prefixes, and the cycle-duration bounds, it certifies attraction throughout
physical execution. Shadow-path comparison gives a practical tube under
bounded leakage. Both branches include execution costs.

For invariance of every target-face state injected at every phase, each
complete physical operation must preserve the target face throughout its path. The two qubit obstructions separate compatibility of the
callable generators from compatibility of the paths implementing conjugation:
a stabilizing average or an algebraic orbit atom alone is insufficient.

The common-ground qutrit realizes the fixed-metric construction. The
four-level network uses a survival observable to enter the same
cycle-and-prefix acceptance test. Its finite-library search returns a
five-command sequence with a preparation-time bound of $14.3107$ in inverse
nominal cascade-rate units and trace-distance tolerance $0.1$ for every
initial state, with this tolerance maintained throughout subsequent operation.

\begin{samepage}
\section*{Data Availability}
The numerical data, exact certificates, and source code supporting the
computational results will be deposited in a public repository before
publication. Pending deposit, these materials are available from the authors
upon reasonable request.
\end{samepage}

\section*{Acknowledgments}
Generative AI tools, including OpenAI ChatGPT and Codex, supported code development and verification, manuscript formatting, and revisions.
Responsibility for the scientific claims, computations, and citations rests
with the authors.

\appendix
\section{Static existence, finite certificates, and presentation costs}
\label{app:static-proofs}

\subsection{Effective cones and geometric information}
The distinction between infinitesimal directions and executable controls also
appears in geometric control.  For the declared
set $G$ of infinitesimal generators, the Lie wedge $W(G)$ is the smallest closed
convex cone containing $G$ and invariant under the adjoint conjugations
$\Ad_gX:=gXg^{-1}$ generated by its edge $E(W)=W\cap(-W)$. For an accessible
right-invariant system on a connected Lie group, the tangent wedge of the
closed reachable semigroup gives its infinitesimal directions. This is the Lie
saturate, the smallest Lie wedge containing the control directions that is
global in the chosen ambient group, meaning that it is realized as the tangent
wedge of a closed semigroup there
\cite{dirr2009lie}.  Without such a globality argument, a locally generated
wedge need not be identified with the tangent wedge.  In coherently controlled
unital Markovian systems, iterative edge, adjoint-invariance, and convex-closure
steps give an inner approximation to the global Lie wedge
\cite{omeara2012unital}.

Let $E_H$ be the real space of directly available Hamiltonian commutators and
let $S$ be the compact set of declared dissipative orbit atoms.  The static
effective cone is
\begin{equation}
 W_{\rm eff}=E_H+\overline{\cone(S)}.
\label{eq:effective-cone}
\end{equation}

\subsection{Presentation-relative costs}
A wedge alone determines neither execution time nor the cost of calling a
generator sum, even before target compatibility is imposed.

\begin{theorem}[Wedge and cost separation]
\label{thm:presentation-cost}
Consider the following finite-dimensional GKLS benchmark presentations with
coherent control restricted to the identity operation, amplitudes in $[0,1]$,
and exclusive sequential access. Their labels can be concatenated at arbitrary
positive dwell times, with no extra time for transitions, dead time, or cooldown.
\begin{enumerate}
\item Let $\mathcal D=\Diss{|0\rangle\langle1|}$ and let the sole maximum
generator be $\gamma\mathcal D$, $\gamma>0$.  Every such presentation has the
same wedge $\mathbb R_{\ge0}\mathcal D$, while the time needed to move
$|1\rangle\langle1|$ within trace distance $\delta\in(0,1)$ of
$|0\rangle\langle0|$ is $\log(1/\delta)/\gamma$.
\item Fix $N\ge2$ and $T>0$, set $d=N+1$, and define
$\mathcal D_j=\Diss{|0\rangle\langle j|}$ and
$X_\star=\sum_{j=1}^N\mathcal D_j$. Let $N_{\rm seg}$ count maximal positive-duration
intervals with one active label, and let $N_{\rm sw}^{\rm one\text{-}shot}$
count changes between successive active labels within the word.
A library with the $N$ separate unit-rate labels and a library augmented by
the unit-rate label $X_\star$ have the same wedge and the same untimed
finite-product semigroup. Nevertheless, every separate-label word whose
channel equals $e^{TX_\star}$ satisfies
\begin{equation}
 T_{\rm wall}\ge NT,\qquad N_{\rm seg}\ge N,
 \qquad N_{\rm sw}^{\rm one\text{-}shot}\ge N-1,
\label{eq:presentation-cost}
\end{equation}
 and all bounds are attained.  Cyclic repetition of a reduced separate-label
word has at least $N$ switches per period. The augmented presentation uses
one active-label call with wall time $T$ and
$N_{\rm sw}^{\rm one\text{-}shot}=0$.
\end{enumerate}
\end{theorem}

In part (i), the trace distance equals the excited population
$e^{-\gamma\int_0^t u(s)\,ds}\ge e^{-\gamma t}$ for $u(s)\in[0,1]$.
A unit-amplitude dwell attains $\delta$ at $t=\log(1/\delta)/\gamma$.
In part (ii), the $\mathcal D_j$ commute and are linearly independent.
Let $s_j$ be the total exposure over all dwells of label $j$ in the word.
An exact match with $e^{TX_\star}$ on each $|j\rangle\langle j|$ forces
$s_j=T$ for every label.
Exclusivity and the amplitude cap give Eq.~\eqref{eq:presentation-cost},
attained by one unit-amplitude dwell of duration $T$ per label. Adding
$X_\star$ preserves the cone and untimed semigroup and makes the same channel
callable in one such dwell. If a separate-label word is within half diamond
distance $\eps$ of $e^{TX_\star}$ and $e^{-T}+\eps<1$, the same test gives
\begin{equation}
 T_{\rm wall}\ge -N\log(e^{-T}+\eps).
\label{eq:approx-presentation-cost}
\end{equation}
Under the same tolerance condition, every label has positive exposure,
so $N_{\rm seg}\ge N$ and
$N_{\rm sw}^{\rm one\text{-}shot}\ge N-1$ still hold.

\subsection{Finite algebraic certificates}
The static certificate combines finite support with a stability test. Let
$\mathbb A=\mathbb R_{\rm alg}$ be the real algebraic numbers, and represent
Hermiticity-preserving trace-annihilating maps in a fixed Hilbert--Schmidt
orthonormal Hermitian basis with entries in $\mathbb A+i\mathbb A$.
Denote their real vector space by $\G_d$.

\begin{definition}[Positive-gauge presentation]
\label{def:positive-gauge}
Such a presentation gives an algebraic target projector $P$, an algebraic
basis of $E_H$, and finitely many algebraic GKLS generators $L_a$. It includes
compact basic semialgebraic sets $K_a$ of unitaries, specified by finitely many
simultaneous polynomial equalities and inequalities over $\mathbb A$, and an
algebraic linear functional $\phi$ satisfying
\begin{equation}
 \phi(e)=0\quad(e\in E_H),\qquad
 \phi(g)\ge\mu>0\quad(g\in S),
\label{eq:positive-gauge}
\end{equation}
where $S=\bigcup_a\{\Ad_uL_a:u\in K_a\}$.  With
$\pi:\G_d\to\G_d/E_H$, define
$s=\dim\spann\pi(S)$.
\end{definition}

Conic Carath\'eodory reduction in the quotient first bounds the support,
after which the gauge bounds conic coefficient mass along convergent sequences of
generators. With both bounds in place, compactness of the atom set gives
closure and an attained representation. The finite positive-atomization step
is related to that in cubature theory \cite{bayer2006tchakaloff}.

\begin{proposition}[Closure and attained support]
\label{prop:attainment}
Every such presentation satisfies
$E_H+\overline{\cone(S)}=E_H+\cone(S)$. Every $L\in W_{\rm eff}$ has a
same-generator representation with labels $a_j$ from the declared list:
\begin{equation}
 L=e+\sum_{j=1}^m\lambda_j\Ad_{u_j}L_{a_j},\qquad
 e\in E_H,\quad u_j\in K_{a_j},\quad \lambda_j>0,\quad m\le s.
\label{eq:attained-generator}
\end{equation}
\end{proposition}

If $k$ fixed linear metadata coordinates are appended to each atom before
support reduction, the same argument preserves their aggregate values together
with the complete generator and its quotient data using at most $s+k$ atoms.
Preserving nonlinear resource costs requires an additional certificate and
the inclusion of the required metadata as additional linear coordinates before
reduction.

The no-leakage equation in Sec.~\ref{sec:static} makes the transverse
matrix well defined, while the conditioned Lyapunov estimate controls its
decay and prefactor. Attained support turns these stability conditions into
the following finite static certificate.
\begin{theorem}[Finite exact static certificate]
\label{thm:finite-static-certificate}
For a presentation satisfying Definition~\ref{def:positive-gauge}, the
following are equivalent:
\begin{enumerate}
\item some $L\in W_{\rm eff}$ makes $\F_P$ invariant and GAS.
\item there are algebraic data satisfying
Eqs.~\eqref{eq:attained-generator}, \eqref{eq:static-no-leakage}, and \eqref{eq:static-lyapunov}, including the
complete generator, its transverse matrix computed from the algebraic data,
positive Routh--Hurwitz determinants, and a strict Lyapunov metric, with
$1\le m\le s\le d^4-d^2$.
\end{enumerate}
\end{theorem}

The proof combines Propositions~\ref{prop:attainment}
and~\ref{prop:fixed-generator-gas} with transfer between real closed fields
in three steps: Sec.~\ref{app:positive-gauge-attainment}
reduces support before taking limits, Sec.~\ref{app:quotient-stability}
relates attraction to transverse decay, and
Sec.~\ref{app:algebraic-equivalence} obtains algebraic witnesses for the
resulting finite conditions.

\begin{corollary}[Finite-input decidability]
\label{cor:static-decision}
For the finite semialgebraic input of Theorem~\ref{thm:finite-static-certificate},
existence of a stabilizing generator in $W_{\rm eff}$ is decidable by real
quantifier elimination.
\end{corollary}

\begin{corollary}[Tangent-wedge interpretation]
\label{cor:tangent-interpretation}
If the tangent wedge of the closed reachable semigroup is independently proved
to equal $W_{\rm eff}$, the conditions of
Theorem~\ref{thm:finite-static-certificate} are also equivalent to existence of
a stabilizing generator in that tangent wedge.
\end{corollary}

\subsection{Positive-gauge attainment}
\label{app:positive-gauge-attainment}

To prove Proposition~\ref{prop:attainment}, let
\begin{equation}
 L_k=e_k+\sum_j\lambda_{kj}g_{kj}\longrightarrow L,\qquad
 e_k\in E_H,\qquad g_{kj}\in S,\qquad\lambda_{kj}\ge0.
\label{eq:static-approximating-sequence}
\end{equation}
Apply the quotient map $\pi$ and conic Carath\'eodory elimination to each
ordinary finite sum before taking a limit, leaving at most
$s=\dim\spann\pi(S)$ nonzero atoms. The gauge then bounds their total
coefficient mass:
\begin{equation}
 \mu\sum_j\lambda_{kj}
 \le\sum_j\lambda_{kj}\phi(g_{kj})=\phi(L_k),
\label{eq:gauge-mass-bound}
\end{equation}
so the total coefficient mass is bounded.  After padding the reduced sums
with zero coefficients, compactness of $S$ and the coefficient interval gives
a convergent subsequence.  After zero limiting coefficients are discarded,
\begin{equation}
 \pi(L)=\sum_{j=1}^m\lambda_j\pi(g_j),\qquad m\le s.
\end{equation}
The remainder $L-\sum_j\lambda_jg_j$ lies in $E_H$, which proves both closure
and the same-$L$ representation.

If $L$ and the presentation are algebraic, the fiber of all such
decompositions is a nonempty semialgebraic set over $\mathbb A$, so transfer
between real closed fields gives an algebraic point in the same
fiber~\cite[Theorem~2.80]{basu2006algorithms}.  If
$L$ is transcendental, real attainment still holds but an algebraic
same-generator representation need not.

\begin{example}[A nonclosed conical hull]
The compact semialgebraic set
\begin{equation}
 S_0=\{(-1,0,0)\}\cup
 \{(1,t,t^2):0\le t\le1\}
\label{eq:compact-nonclosed-cone-set}
\end{equation}
has a nonclosed cone.  Indeed,
$n(1,1/n,1/n^2)+n(-1,0,0)=(0,1,1/n)$ belongs to
$\cone(S_0)$ and converges to $(0,1,0)$.  In any conic representation with
third coordinate zero, every active curved atom has $t=0$, which also makes
the second coordinate zero.  Thus $(0,1,0)\notin\cone(S_0)$.  The positive
gauge excludes precisely this escape of coefficient mass.
\end{example}

\subsection{Quotient stability and the conditioned bound}
\label{app:quotient-stability}

With Eq.~\eqref{eq:static-no-leakage}, $e^{tL}$ preserves $\T_P$ and hence
$\F_P$, and induces $e^{tA_L}$ on $\X_P$. If $A_L$ is Hurwitz, every quotient
class decays exponentially, so every state approaches $\F_P$.
Conversely, invariance of $\F_P$ gives invariance of its real span $\T_P$.
Differentiating at $t=0$ yields Eq.~\eqref{eq:static-no-leakage}.
Since density operators span $\Herm(\mathcal H)$ over $\mathbb R$, attraction
of every state gives $e^{tA_L}z(\rho)\to0$ as $t\to\infty$ on a spanning set
of quotient coordinates. This excludes eigenvalues with nonnegative real
part and proves Proposition~\ref{prop:fixed-generator-gas}.
The Routh--Hurwitz criterion supplies the corresponding polynomial form.

For $V(x)=x^{\mathsf T}Mx$, Eq.~\eqref{eq:static-lyapunov} implies
$\dot V\le-2\nu V$. Comparing with the two metric bounds converts this
decay to the fixed quotient coordinates, retaining the conditioning factor:
\begin{equation}
 \|x(t)\|_2\le\sqrt{m_+/m_-}\,e^{-\nu t}\|x(0)\|_2.
\label{eq:appendix-lyapunov-bound}
\end{equation}
Let $Q_\rho=\rho-P\rho P$ represent the quotient class orthogonally in the
Hilbert--Schmidt inner product, and put $p=\tr(P\rho)$. For $p>0$, $\sigma=P\rho P/p$ belongs to
$\F_P$ and $\|\rho-\sigma\|_1\le\|Q_\rho\|_1+1-p\le2\|Q_\rho\|_1$, since
$1-p=\tr Q_\rho$. For $p=0$, $Q_\rho=\rho$ and $\|\rho-\sigma\|_1\le2=2\|Q_\rho\|_1$
for any $\sigma\in\F_P$. The norm comparisons then give
\begin{equation}
 \dtr(\rho,\F_P)\le\|Q_\rho\|_1
 \le\sqrt d\,\|Q_\rho\|_{\rm HS},\qquad
 \|Q_{\rho(0)}\|_{\rm HS}\le2.
\label{eq:trace-hs-conversion}
\end{equation}
Combining Eqs.~\eqref{eq:appendix-lyapunov-bound} and
\eqref{eq:trace-hs-conversion} proves Eq.~\eqref{eq:conditioned-static-bound}.

No-leakage and transverse stability are independent.  For a qubit, adding an
upward jump to downward amplitude damping can leave a deleted-coordinate
block Hurwitz while $P_0$ leaks.  Conversely, a generator can preserve the
target face while leaving a transverse mode undamped.  Both parts of
Proposition~\ref{prop:fixed-generator-gas} are therefore necessary.

The prefactor cannot in general be recovered from the spectral abscissa
alone. A finite-time example is given by the abstract linear-algebra family
\begin{equation}
 A_K=\begin{pmatrix}-1&K\\0&-1\end{pmatrix},\qquad K>0,
\label{eq:nonnormal-family}
\end{equation}
which has spectral abscissa $-1$ for every $K$, but
$e^{A_K}=e^{-1}\begin{psmallmatrix}1&K\\0&1\end{psmallmatrix}$, whose Euclidean
norm is unbounded as $K\to\infty$.  The Lyapunov estimate keeps both the
decay rate and the metric-dependent prefactor.

\subsection{Algebraic equivalence and decision scope}
\label{app:algebraic-equivalence}

Sufficiency in Theorem~\ref{thm:finite-static-certificate} follows immediately
from algebraic reconstruction of $L$, no leakage, positive Routh--Hurwitz data,
and the strict Lyapunov inequality.

For necessity, a real stabilizing $L$ has attained support of size $m\le s$.
Since a Hamiltonian generator fixes $I/d$, $m\ge1$. Projecting the algebraic
Hermitian basis by $X\mapsto X-PXP$ and orthonormalizing gives an algebraic basis
of $\T_P^\perp$. In these coordinates the orbit constraints, generator and
block equalities, no leakage, and strict weight and Routh--Hurwitz inequalities
are semialgebraic over $\mathbb A$. Real-closed-field transfer gives algebraic
generator data.

For the resulting algebraic Hurwitz matrix, an algebraic positive-definite
solution exists for
\begin{equation}
 A_L^{\mathsf T}M+MA_L=-I.
\label{eq:algebraic-lyapunov-equation}
\end{equation}
Strictness permits algebraic choices of $m_-,m_+,\nu$, and $q_0$. This
completes the necessity proof.

Including the Lyapunov matrix inequalities gives the full certificate as a
finite existential semialgebraic sentence over the same field. Quantifier
elimination decides this sentence~\cite[Theorem~2.77 and Algorithm~11.14]{basu2006algorithms}.
Compactness and the universal gauge condition can themselves be
stated over the reals.

The support cap is presentation-relative.  The commuting common-ground sum
needs two atoms in a library containing only its two branches but one call in
a library that directly supplies their sum.

\subsection{Support on compatible callable atoms}
The support bound for executable atoms that preserve the target uses the
quotient span dimension recomputed for the restricted library and its coherent
edge. Provided the hypotheses hold for these restricted data, the same support
argument applies.

\begin{proposition}[Compatible finite support]
\label{prop:compatible-attainment}
Let
$E_{H,P}=E_H\cap\ker N_P$ and let $S_{\rm tc}$ be the full set of executable
complete atoms that preserve $\F_P$.  Suppose $S_{\rm tc}$ is compact
semialgebraic with algebraic defining data, an algebraic linear gauge $\phi$
annihilates $E_{H,P}$ and
satisfies $\phi(G)\ge\mu>0$ on $S_{\rm tc}$, and
\[
 L\in E_{H,P}+\overline{\cone(S_{\rm tc})}.
\]
Then $E_{H,P}+\cone(S_{\rm tc})$ is closed, and $L$ has finite compatible
support of size at most
\[
 s_{\rm tc}=\dim\spann\pi_{\rm tc}(S_{\rm tc}),
\]
where $\pi_{\rm tc}$ is the quotient by $E_{H,P}$. If $L$, $E_{H,P}$, $S_{\rm tc}$,
the gauge, and the calibration conditions defining the executable atoms are
defined over $\mathbb A$, the support can be chosen algebraic.
\end{proposition}

\begin{proof}
Applying the positive-gauge argument in Sec.~\ref{app:positive-gauge-attainment} with
$E_H,S,\pi$ replaced by $E_{H,P},S_{\rm tc},\pi_{\rm tc}$ and span dimension
$s_{\rm tc}=\dim\spann\pi_{\rm tc}(S_{\rm tc})$ gives an attained decomposition
that makes the real decomposition fiber nonempty, with each atom satisfying its
calibration conditions. For algebraic input, the real-closed-field transfer
step in Sec.~\ref{app:algebraic-equivalence} supplies algebraic witnesses in
this fiber.
\end{proof}

\section{Target compatibility and qubit obstructions}
\label{app:target-proofs}

\subsection{Face-invariance criteria}

To prove Lemma~\ref{lem:gkls-face}, first test leakage through the
complementary population block, then test the off-diagonal block. For a
target-supported positive operator
$\rho=P\rho P$, the $QQ$ block of the GKLS derivative is
\begin{equation}
 Q\mathcal L(\rho)Q
 =\sum_a(QV_aP)\rho(QV_aP)^\dagger.
\label{eq:gkls-qq-block}
\end{equation}
Every term is positive semidefinite.  If the target face is invariant for all
target-supported states, the sum vanishes for all such $\rho$, and hence
$QV_aP=0$ for every $a$.  Under these conditions the $QP$ block becomes
\begin{equation}
 Q\mathcal L(\rho)P
 =-iQ\left(H-\frac{i}{2}R\right)P\rho P,\qquad
 R=\sum_aV_a^\dagger V_a.
\label{eq:gkls-qp-block}
\end{equation}
Its vanishing for every target-supported $\rho$ is equivalent to the second
condition in Eq.~\eqref{eq:gkls-face-criterion}, with the adjoint $PQ$ block then
vanishing as well.  Conversely, the two conditions make both transverse blocks
zero and leave the derivative in $\T_P$, proving invariance.

For a CPTP map $\Phi(\rho)=\sum_aK_a\rho K_a^\dagger$, the corresponding
calculation is
\begin{equation}
 Q\Phi(\rho)Q=\sum_a(QK_aP)\rho(QK_aP)^\dagger.
\label{eq:cptp-qq-block}
\end{equation}
This proves Eq.~\eqref{eq:cptp-face-criterion}.  For a bounded
piecewise-regular time-dependent GKLS equation, applying the infinitesimal
calculation to every target-face state at every injection phase proves
necessity almost everywhere, while
integrating a generator that satisfies the conditions proves sufficiency. For a single trajectory in a model allowing a noninvertible channel
or reset segment, that earlier segment may already have collapsed the reachable
target set, so the same necessity statement need not apply.

With no jumps, Eq.~\eqref{eq:gkls-face-criterion} reduces to $QH(s)P=0$ and
its adjoint.  The Schr\"odinger equation then preserves both $P\mathcal H$ and
$Q\mathcal H$, so $U(s)\in\Stab(P)$.  Differentiation proves the converse.
This establishes Eq.~\eqref{eq:unitary-stabilizer-path}.

The complete generator need not split into separately target-preserving
Hamiltonian and dissipative parts.  For example, with
\begin{equation}
 V=|0\rangle(\langle0|+\langle1|),\qquad H=Y/2,
\label{eq:cross-block-compensation-example}
\end{equation}
the two cross-block contributions cancel in the second condition of
Eq.~\eqref{eq:gkls-face-criterion}, although $[H,P_0]\ne0$.  By contrast,
positive jump leakage in Eq.~\eqref{eq:gkls-qq-block} admits no Hamiltonian
cancellation.  Compensation must therefore be solved for the complete simultaneous
operating point.

An arbitrarily small general GKLS perturbation can destroy exact target
invariance. For the qubit target $P_0$, let
\begin{equation}
 V_\epsilon=\sqrt\epsilon\,|1\rangle\langle0|,
 \qquad \epsilon>0.
\label{eq:leakage-jump-example}
\end{equation}
Then $QV_\epsilon P_0\ne0$, and the solution of
$\dot\rho=\Diss{V_\epsilon}(\rho)$ from $P_0$ has excited population
$1-e^{-\epsilon t}$.  The rate can be arbitrarily small while the target face
fails to be invariant for every $t>0$.  Hence exact invariance is not open in
an unrestricted generator-norm ball, which motivates the separate shadow-path
conclusion of Theorem~\ref{thm:primitive-practical-tube}.

\subsection{Strict positivity under finite Markovian propagation}

The stable-average obstruction and the always-on-pumping part of the orbit
obstruction use the following persistence property: once a segment produces
a full-rank state, later bounded finite Markovian segments cannot return it
to a pure target.

\begin{lemma}[Interior preservation]
\label{lem:interior-preservation}
Let $\Phi$ be the finite-time propagator of a bounded finite-dimensional GKLS
equation.  If $\rho\succ0$, then $\Phi(\rho)\succ0$.
\end{lemma}

\begin{proof}
The fundamental solution of the linear master equation is invertible as a
linear map, and the GKLS theorem makes it CPTP.  Write
$\Phi(X)=\sum_aK_aXK_a^\dagger$.  If $\Phi(I)$ had a nonzero kernel vector
$v$, then
\begin{equation*}
 0=\langle v|\Phi(I)|v\rangle=\sum_a\|K_a^\dagger v\|^2
\end{equation*}
would imply $K_a^\dagger v=0$ for all $a$.  Every matrix in the range of
$\Phi$ would then annihilate $v$ on the right, contradicting linear
invertibility.  Hence $\Phi(I)\succ0$.  Since $\rho\succeq cI$ for some
$c>0$, complete positivity gives
$\Phi(\rho)\succeq c\Phi(I)\succ0$.
\end{proof}

\subsection{The stable-average qubit}

To test what an individual hold does to the pure target, use the Bloch
convention $\rho=(I+xX+yY+zZ)/2$, with $P_0$ at $(0,0,1)$. Under either
generator in Eq.~\eqref{eq:qubit-average-atoms}, the Hamiltonian part is a
rotation and drops out of the derivative of $x^2+y^2+z^2$. The damping part gives
\begin{equation}
 \frac{d}{dt}(x^2+y^2+z^2)=-x^2-y^2+2z(1-z),
\end{equation}
which is equivalent to Eq.~\eqref{eq:purity-defect}.  Its integrating-factor
form is
\begin{equation}
 \delta(t)=e^{-t}\left[\delta(0)+
 \int_0^t e^s[1-z(s)]^2\,ds\right].
\label{eq:purity-defect-integral}
\end{equation}
For a trajectory starting at $P_0$, the transverse Hamiltonian gives a nonzero $y$ derivative,
which would contradict the Bloch equations if $z$ remained identically one on
any interval.  The integrand in
Eq.~\eqref{eq:purity-defect-integral} is therefore positive on a subinterval
of every positive dwell, and $\delta(t)>0$ for every $t>0$.  Because a qubit state has
$\delta>0$ if and only if it is in the full-rank interior, Lemma
\ref{lem:interior-preservation} prevents all later finite GKLS factors from
returning it to $P_0$, proving the nonreturn claim of
Theorem~\ref{thm:static-word-no-go}.

To construct the average generator's static certificate, use the quotient
convention of Sec.~\ref{sec:presentations} with the ordered Hilbert--Schmidt
orthonormal basis $(P_1,X/\sqrt2,Y/\sqrt2)$, where $P_1=|1\rangle\langle1|$.
The average $(\mathcal L_++\mathcal L_-)/2=\mathcal D$ has transverse matrix
$A_{\mathcal D}=\operatorname{diag}(-1,-1/2,-1/2)$ and characteristic
polynomial
\begin{equation}
 z^3+2z^2+\frac54z+\frac14.
\label{eq:qubit-static-characteristic-polynomial}
\end{equation}
Its three leading Routh--Hurwitz determinants are respectively
$2$, $9/4$, and $9/16$.  Here $K_2=\sqrt{m_+/m_-}$ is the Euclidean
conditioning factor of the Lyapunov metric and
$K_{\rm tr}=2\sqrt d\,K_2$ is the resulting uniform trace-distance
prefactor.  One rational static certificate is
\begin{equation}
 M=I_3,\qquad \nu=\tfrac14,\qquad m_-=\tfrac12,\qquad m_+=2,
 \qquad q_0=\tfrac14,\qquad K_2=2,\qquad K_{\rm tr}=4\sqrt2.
\label{eq:qubit-static-receipt}
\end{equation}
The presentation uses weights $1/2,1/2$, zero coherent edge, and the gauge
$\phi=-\operatorname{Tr}_{\rm superop}$, for which
$\phi(\mathcal L_\pm)=2$.  Thus the two-atom set is compact and satisfies the
positive-gauge condition.  These data complete the static certificate for
the target-compatible average and, together with the obstruction, show that
a target-compatible sum need not admit a decomposition into target-compatible
primitive atoms.

\subsection{The orbit-to-word qubit}
\label{subsec:orbit-word-proof}

The endpoint and the pulse path establish different parts of the
obstruction. The endpoint identity follows because conjugation by $X$ interchanges
$|1\rangle\langle0|$ and $|0\rangle\langle1|$, proving
Eq.~\eqref{eq:qubit-stroboscopic-sandwich}. For the pulse path
$U(s)=e^{-is\Omega X/2}$,
\begin{equation}
 U(s)|0\rangle
 =\cos(\Omega s/2)|0\rangle-i\sin(\Omega s/2)|1\rangle.
\end{equation}
The trace distance between two pure states is the square root of one minus
their squared overlap, giving Eq.~\eqref{eq:qubit-pulse-excursion}.  Both pulse
directions and $D_\uparrow$ have a nonzero derivative at $P_0$. The first
positive-duration segment therefore leaves the target in any word of positive
total duration, proving the all-time part of Theorem~\ref{thm:orbit-word-no-go}.

When upward pumping remains on, the path obstruction strengthens to strict
positivity. Let
\begin{equation}
 K_u=\gamma D_\uparrow-i[uX/2,\,\cdot\,],\qquad \gamma>0
\end{equation}
with bounded $u$.  Direct expansion from $P_0$ gives
\begin{equation}
 \det\rho(s)=\gamma s+O(s^2)>0
\label{eq:always-on-determinant}
\end{equation}
for all sufficiently small $s>0$.  Lemma~\ref{lem:interior-preservation}
propagates strict positivity through every later bounded finite segment.

\section{Primitive-path, product, timing, and resource bounds}
\label{app:monodromy}

The derivations below support the primitive compiler of
Sec.~\ref{sec:physical-words} and the direct-label schedules of
Appendix~\ref{sec:direct}. We first derive the path and product bounds, then
convert contraction into timing and resource guarantees. The final subsection
separates attraction to a face from selection of a state within that face.

\subsection{Primitive-path and nominal-product enclosures}

To compare two complete primitive paths on one clock, let $\Phi(u,v)$ and
$\widehat\Phi(u,v)$ be the propagators of $G$ and
$\widehat G$ on normalized local time. Differentiating
$\Phi(1,u)\widehat\Phi(u,0)$ gives
\begin{equation}
 \frac{d}{du}\bigl[\Phi(1,u)\widehat\Phi(u,0)\bigr]
 =\Phi(1,u)(\widehat G(u)-G(u))\widehat\Phi(u,0).
\label{eq:duhamel-derivative}
\end{equation}
Integration proves Eq.~\eqref{eq:primitive-duhamel-identity}.  The definition
of the logarithmic norm implies
\begin{equation}
 \|\Phi(v,u)\|_M\le
 \exp\!\left(\int_u^v\mu_M(G(s))\,ds\right),\qquad 0\le u\le v\le1,
\label{eq:log-norm-propagator-bound}
\end{equation}
because the derivative of $\|\Phi(s,u)x\|_M^2$ is bounded by twice the
logarithmic norm multiplied by that squared norm.  Applying Eq.~\eqref{eq:log-norm-propagator-bound} on both
sides of $G-\widehat G$ proves every assertion of
Lemma~\ref{lem:primitive-duhamel}.  The normalized clocks are essential: their
derivatives multiply the physical generators, so duration calibration is part
of the same difference as waveform and amplitude error.

For the nominal-product enclosure, the exact factor telescope is
\begin{equation}
 R_N\cdots R_1-W_N\cdots W_1
 ={}\sum_{k=1}^NR_N\cdots R_{k+1}
 {}\times(R_k-W_k)W_{k-1}\cdots W_1.
\label{eq:nominal-factor-telescope}
\end{equation}
Using this telescope with the factors surrounding each local error retained,
combine bounds on $R_k-W_k$, the left reference factors, and
the right ideal factors with the ideal-product error to obtain the reference-to-static estimate in Eq.~\eqref{eq:derived-nominal-synthesis}.
Applying the same identity to the actual--reference errors $C_k-R_k$ gives
Eq.~\eqref{eq:robust-word-telescope}, with the full amplification in
Eq.~\eqref{eq:robust-word-prefix-bounds}.

For direct labels, Appendix~\ref{sec:direct} supplies a finite estimate of
$E_{\rm split}$. Use the transverse atoms $B_j$, the nonnegative norm budgets
$b_j,c_{kj}\ge0$, and $\beta,C_2$ defined before
Eq.~\eqref{eq:product-error-bounds}.  Exchanging the exact exponential with the
ordered product one factor at a time gives the first-order bound in that
equation. Since the symmetric product in Eq.~\eqref{eq:strang-word} matches the
exact exponential through second order, bounding the two remainders from
third order onward gives the symmetric bound in
Eq.~\eqref{eq:product-error-bounds}.
These estimates use the induced-$M$ norm. Trace- or diamond-norm versions
require explicit norm-conversion bounds.

For finite rational upper bounds, let $p\ge0$ be an integer. If $x\ge0$
and $x<p+2$, then
\begin{equation}
 U_p(x)=\sum_{k=0}^p\frac{x^k}{k!}
 +\frac{x^{p+1}}{(p+1)!}\frac{1}{1-x/(p+2)}
\label{eq:upper-exp-majorant}
\end{equation}
satisfies $e^x\le U_p(x)$.  For $y>0$,
\begin{equation}
 e^{-y}<V_p(y)=\left(\sum_{k=0}^p\frac{y^k}{k!}\right)^{-1}.
\label{eq:negative-exp-majorant}
\end{equation}
A geometric-series estimate bounds the remainder in
Eq.~\eqref{eq:upper-exp-majorant}. Positivity of the omitted exponential terms
gives Eq.~\eqref{eq:negative-exp-majorant}. These functions turn the displayed
sufficient product bounds into finite rational inequalities.

\subsection{Ordered interstitial maps and monodromy iteration}

Let $E_1,\ldots,E_p$ be the induced transverse maps of the chronological
active factors, and let $D_\ell$ be the induced transverse map of the
interstitial channel following $E_\ell$.  With $Q_\ell=E_p\cdots E_{\ell+1}$,
repeated insertion of $Q_\ell Q_\ell^{-1}$ gives
Eq.~\eqref{eq:interstitial-factorization}.
If $\|D_\ell-I\|_M\le\Delta_\ell$ and the active suffix and inverse suffix
have norms at most $e^{hs_\ell}$, submultiplicativity gives
Eq.~\eqref{eq:interstitial-bound}.  The suffix conjugations cannot be removed
without an additional commutation argument.

The same cycle-and-prefix logic also applies when period maps are certified
directly. For path-compatible cold cycles, Proposition~\ref{prop:finite-search-test}
applies with $V(\rho)=\|z(\rho)\|_M$, $\psi(v)=K_{\rm out}v$,
$Y=K_{\rm in}\dtr(\rho(0),\F_P)$, and $\eta=b=0$.
A common cycle contraction $0\le q<1$ and prefix bound $B$ then give the
all-time estimate in Eq.~\eqref{eq:robust-all-time-bound}, including
cycle-varying admitted uncertainty.

The next two examples explain why the preceding bounds use finite products
and a common contraction estimate. Stability of the average and stability of
individual maps, respectively, do not suffice. Let $E_{ij}=e_ie_j^{\mathsf T}$ be the $2\times2$ matrix units
in the standard basis $e_1,e_2$ of $\mathbb R^2$.  First, set
\begin{equation}
 G_1=-I+KE_{12},\qquad G_2=-I-KE_{21},\qquad K>0.
\label{eq:coarse-word-generators}
\end{equation}
Their average $A_L=(G_1+G_2)/2$ satisfies
$A_L^\mathsf{T}+A_L=-2I$, but at $K=10$ the unit-exposure word obeys
\begin{equation}
 \left\|e^{G_2/2}e^{G_1/2}\right\|_2>9.
\label{eq:coarse-word-amplification}
\end{equation}
Indeed, its action on the second basis vector has norm
$e^{-1}\sqrt{601}>9$.  A static Hurwitz margin therefore does not control a
coarse finite word without commutator, path, and amplification bounds.

Second, consider the following abstract matrices
\begin{equation}
 A=\begin{pmatrix}1/2&1\\0&1/2\end{pmatrix},\qquad
 B=\begin{pmatrix}1/2&0\\1&1/2\end{pmatrix}
\label{eq:pointwise-product-example}
\end{equation}
which each have spectral radius $1/2$, whereas
$\rho(BA)=(3+2\sqrt2)/4>1$, giving a counterexample to pointwise
spectral-radius reasoning.

The physical examples establish separate limitations:
Eq.~\eqref{eq:leakage-jump-example} shows loss of exact target invariance under
arbitrarily small GKLS perturbations, while the gain box containing zero in
Appendix~\ref{app:qutrit-calculations} precludes uniform strict contraction.

\subsection{Timing domains and resource branches}

The direct schedule converts minimum dwell times into a lower exposure
endpoint: requiring $t_j=\lambda_jh/r_j\ge\tau_j$ in the first-order schedule
gives Eq.~\eqref{eq:first-order-dwell}.  In the
cut-open symmetric schedule the central label appears once with full exposure,
whereas the remaining labels appear in two half arcs, so
\begin{equation}
 h_{\min,2}^{\rm dwell}=\max\left(
 \left\{\frac{r_{\pi_m}\tau_{\pi_m}}{\lambda_{\pi_m}}\right\}
 \cup\left\{\frac{2r_{\pi_k}\tau_{\pi_k}}{\lambda_{\pi_k}}:
 1\le k<m\right\}\right).
\label{eq:strang-dwell-lower}
\end{equation}
Fixed-duration transitions still contribute to the physical cycle and actual
monodromy.  Pulse amplitude, slew, cooldown, and calibration constraints may
add further positive lower endpoints, all of which enter the maximum
in Eq.~\eqref{eq:physical-hmin}. Finite support alone gives no finite uniform
upper bound on the required minimum exposure, as shown by
$h_{\min}^{\rm dwell}=\tau/\delta$ in Eq.~\eqref{eq:small-weight-dwell}.

To certify contraction throughout an open interval $h_-<h<h_+$ for the direct
schedule, use the static-decay and product-error envelopes of
Eq.~\eqref{eq:product-bounds-summary}.  Sufficient choices for the latter are
given in Eq.~\eqref{eq:product-error-bounds}.  Choose the decay envelope to be
nonincreasing and the error envelope to be nondecreasing between the endpoints.
One sufficient criterion is that there are $0<\theta<1$ and an interstitial bound
$\overline\kappa_+$ that is a uniform upper envelope for the interstitial enclosure in
Eq.~\eqref{eq:interstitial-bound} over the candidate interval, such that
\begin{equation}
 \overline\kappa_+\overline d_\nu(h_-)\le1-\theta,
 \qquad
 \overline\kappa_+\overline E_o(h_+)<\theta.
\label{eq:timing-interval-receipt}
\end{equation}
The endpoint inequalities then imply the contraction condition in
Eq.~\eqref{eq:direct-qbar} throughout the interval.  A pointwise exact
calculation instead certifies only the selected word.

\subsection{Complete cycle-count branches}
With the post-initialization clock and constants $K_{\rm eff}$ and $d_0$
from Sec.~\ref{sec:physical-words}, the following cases give the smallest
nonnegative cycle count certified by Eq.~\eqref{eq:robust-all-time-bound}
in the exact branch for tolerance $\eps\ge0$. They cover zero cycle
contraction and certify zero cycles when $K_{\rm eff}d_0\le\eps$.
\begin{equation}
N_\eps^{\rm ex}=
\begin{cases}
0,&K_{\rm eff}d_0\le\eps,\\
1,&K_{\rm eff}d_0>\eps\ \text{and }q_*=0,\\
\left\lceil\dfrac{\log(K_{\rm eff}d_0/\eps)}{-\log q_*}\right\rceil,
&K_{\rm eff}d_0>\eps>0,\quad 0<q_*<1,\\
+\infty,&\eps=0<K_{\rm eff}d_0,\quad 0<q_*<1.
\end{cases}
\label{eq:cycle-count-branches}
\end{equation}

For the leaky branch, define
\begin{equation}
 D_0=K_{\rm out}(B_*K_{\rm in}d_0+\eta_{\rm pre}^M),\qquad
 D_n=R_\infty+q_*^n(D_0-R_\infty).
\label{eq:leaky-cycle-envelope}
\end{equation}
If $\eps>R_\infty$, the certified count is zero when $D_0\le\eps$, one when
$D_0>\eps$ and $q_*=0$, and, when $D_0>\eps$ and $0<q_*<1$,
\begin{equation}
 N_\eps^{\rm leak}=\left\lceil
 \frac{\log[(D_0-R_\infty)/(\eps-R_\infty)]}{-\log q_*}
 \right\rceil.
\label{eq:leaky-cycle-count}
\end{equation}
For $\eps<R_\infty$, this upper bound does not certify eventual all-time
attainment of the requested tolerance.  At $\eps=R_\infty$, the certified
count is zero if $D_0\le R_\infty$, one if $D_0>R_\infty$ and $q_*=0$,
and infinite if $D_0>R_\infty$ and $0<q_*<1$.  These cases include
$\eps=R_\infty=0$, with $q_*^0=1$ also when $q_*=0$.

In the exact robust branch, Eq.~\eqref{eq:robust-all-time-bound} reduces the
tolerance calculation to $K_{\rm eff}d_0q_*^N\le\eps$, whose four elementary
cases give Eq.~\eqref{eq:cycle-count-branches}.  In the leaky branch,
Eq.~\eqref{eq:primitive-practical-tube} equals
$R_\infty+q_*^n(D_0-R_\infty)$ after substituting
Eqs.~\eqref{eq:practical-radius} and \eqref{eq:leaky-cycle-envelope}.  Solving
this scalar inequality gives Eq.~\eqref{eq:leaky-cycle-count} for
$\eps>R_\infty$.  At $\eps=R_\infty$, the sign of $D_0-R_\infty$
and whether $q_*=0$ give the boundary cases stated there.

For a cold closed cycle, componentwise resources compose as
\begin{equation}
 C(N)=C_{\rm init}+NC_{\rm cyc}+C_{\rm final}.
\label{eq:cold-cycle-account}
\end{equation}
For $N\ge1$ repetitions of a cut-open body with a separate seam,
\begin{equation}
 C(N)=C_{\rm init}+NC_{\rm body}+(N-1)C_{\rm seam}+C_{\rm final}.
\label{eq:seamed-cycle-account}
\end{equation}
For $N=0$, no body or seam is executed, so the cost consists only of any
initialization and finalization that are actually performed.
Switches are counted on the complete physical-signature sequence. Algebraic
cancellation of ideal endpoints does not remove a physical edge.

Fixed per-cycle overhead also changes the fast-switching limit.  If
$T_{\rm cyc}(h)=\tau_0+ah+o(h)$,
$q(h)=1-\nu h+o(h)$, $\tau_0>0$, and $\nu>0$, then
\begin{equation}
 \frac{T_{\rm cyc}(h)}{-\log q(h)}
 \sim\frac{\tau_0}{\nu h},
\label{eq:fast-switching-overhead}
\end{equation}
so the coefficient diverges as $h\downarrow0$.

\subsection{Stroboscopic and intraface limits}

Cycle-boundary attraction determines the within-cycle limit only after the
prefix is applied. For a rank-one target, repeat a fixed physical cycle whose
period map fixes $P$ and contracts boundary deviations. Continuity of its
prefix at a fixed phase $s$ gives
\begin{equation}
 \rho(nT_{\rm cyc}+s)\longrightarrow
 \rho_\star(s)=\Psi(s,0)P.
\label{eq:periodic-orbit-limit}
\end{equation}
This orbit is constant exactly when every prefix fixes $P$.

\paragraph{A higher-rank target face.}
Let
\begin{equation}
 P=P_0+P_1,\qquad
 A_\Lambda=\Diss{|0\rangle\langle2|},\qquad
 B_\Lambda=\Diss{|1\rangle\langle2|}.
\label{eq:lambda-face}
\end{equation}
Both generators preserve $\F_P$ and have the same transverse action, although
$[A_\Lambda,B_\Lambda](P_2)=P_1-P_0$.  A swap within the target subspace
conjugates one jump into the other along a face-preserving path.  The
transverse certificate proves attraction to $\F_P$ but cannot select a unique
mixed state because every state in the face is stationary under a single
decay from level two.  A unique target would require a separate contraction
certificate for motion within the face.

For a higher-rank face, suppose the cycle maps preserve $\F_P$ and the
transverse coordinates of the resulting state trajectories converge to zero
for every admitted cycle sequence.
Fix a selected state $\sigma_\star\in\F_P$.  For the state $\rho_n$ after
$n$ cycles, put $z_n=z(\rho_n)$ and
\[
 Y_n=P\rho_nP-\tr(P\rho_n)\sigma_\star.
\]
This trace-zero operator, defined even when $\tr(P\rho_n)=0$,
has coordinates $y_n^F$ in a fixed real Hermitian basis of
$\{Y\in\T_P:\tr Y=0\}$.  An additional intraface certificate consists of a
fixed positive-definite metric $M_F$ on this coordinate space and constants
$0\le q_F<1$ and $0\le g<\infty$ such that
\begin{equation}
 \|y_{n+1}^F\|_{M_F}\le q_F\|y_n^F\|_{M_F}+g\|z_n\|_M
\label{eq:intraface-certificate}
\end{equation}
holds for every density-operator input and every admitted cycle map, with
the same metric and constants.  Since $z_n\to0$, iteration gives
$y_n^F\to0$.  The identity
\[
 \rho_n-\sigma_\star
 =Y_n+Q_{\rho_n}-\tr(Q_{\rho_n})\sigma_\star
\]
then proves convergence to $\sigma_\star$ at cycle boundaries. For all-time
selection, suppose every admitted physical prefix is CPTP and fixes
$\sigma_\star$. If $\Psi_n(s,0)$ is the prefix of the cycle starting at
$\rho_n$, CPTP contractivity gives
\[
 \dtr(\Psi_n(s,0)\rho_n,\sigma_\star)\le\dtr(\rho_n,\sigma_\star).
\]
The boundary convergence thus extends uniformly through each cycle.

\section{Qutrit calculations}
\label{app:qutrit-calculations}

The calculation first fixes the complete calibrated paths and their reference
monodromies. It then bounds gain errors and intermediate prefixes before
converting the resulting contraction into physical resources.

For $0\le j<k\le2$, define the normalized Hermitian coherences
\begin{equation}
 X_{jk}=\frac{|j\rangle\langle k|+|k\rangle\langle j|}{\sqrt2},
 \qquad
 Y_{jk}=\frac{-i(|j\rangle\langle k|-|k\rangle\langle j|)}{\sqrt2}.
\label{eq:qutrit-hermitian-basis}
\end{equation}
Together with the three populations they form an orthonormal Hermitian basis.

\subsection{Complete paths and reference word}
The controller starts in the laboratory frame with an idle actuator and no
relevant oscillator phase, and the first hold ends in that frame. The forward
pulse moves to the swapped frame, where the second hold ends, before the return
pulse restores the laboratory frame and idle actuator. Both triangular
commands vanish at their endpoints. The environment remains
time-local Markovian and requires no reset. The complete cycle therefore
returns every controller field to its cold-boundary value.  Write $D_{01}$ for the complete physical signature of a raw
$A$-hold and $\mathrm{IDLE}$ for the declared cold signature.  The signature
sequence, with $S_{12}$ and $S_{12}^{-1}$ labeling the separately calibrated forward
and return edges, is
\begin{equation}
 \mathrm{IDLE}\to D_{01}\to S_{12}\to D_{01}\to S_{12}^{-1}
 \to\mathrm{IDLE}.
\label{eq:common-ground-signatures}
\end{equation}
The common-ground calibration interval $h\in(1/2,2)$ applies to each of the
two damping holds, with operating point $h=1$ used here. Each swap edge has
fixed duration one, so all four segments last one unit at this operating
point. Each raw damping envelope is the triangular function
\begin{equation}
 a(s)=
 \begin{cases}
 8\log(5/4)s,&0\le s\le1/2,\\
 8\log(5/4)(1-s),&1/2\le s\le1,
 \end{cases}
\label{eq:damping-ramp}
\end{equation}
with total dissipative exposure
$\alpha=2\log(5/4)$, peak $4\log(5/4)$, and absolute slope at most
$8\log(5/4)$. The forward and return pulse Hamiltonians are,
respectively,
\begin{equation}
 H_+(s)=-u(s)S_{12},\qquad H_-(s)=+u(s)S_{12},
\label{eq:directed-swap-Hamiltonians}
\end{equation}
with the triangular envelope
\begin{equation}
 u(s)=
 \begin{cases}
 2\pi s,&0\le s\le1/2,\\
 2\pi(1-s),&1/2\le s\le1,
 \end{cases}
\label{eq:swap-ramp}
\end{equation}
The pulse envelope has peak $\pi$ and absolute slope at most $2\pi$.
The endpoint unitaries are
\begin{equation}
 U_+(1)=iS_{12},\qquad U_-(1)=-iS_{12},\qquad U_-(1)U_+(1)=I.
\label{eq:directed-swap-endpoints}
\end{equation}
Their channels coincide because the swap is involutory.

The triangular integrals are
\begin{equation}
 \int_0^1a(s)\,ds=2\log(5/4),\qquad
 \int_0^1u(s)\,ds=\frac\pi2.
\label{eq:qutrit-waveform-integrals}
\end{equation}
Self-commutation along each path makes the endpoint depend only on these
integrals.  The swap Hamiltonian is block diagonal with respect to
$\spann\{|0\rangle\}\oplus\spann\{|1\rangle,|2\rangle\}$, so
Eq.~\eqref{eq:unitary-stabilizer-path} holds throughout the pulse.

\begin{proposition}[Reference monodromies of the common-ground word]
\label{prop:common-ground-word}
Consider the four-segment word in Eq.~\eqref{eq:common-ground-word} with the
complete paths in Eq.~\eqref{eq:common-ground-complete-paths} and the
triangular envelopes in Eqs.~\eqref{eq:damping-ramp} and \eqref{eq:swap-ramp}.
Set $h=1$ and $m_1=m_3=1$, and take $P_0$ as the target state. In each of the
two cases below, the word is closed at the cold boundary and path-compatible
with $\F_{P_0}$, realizing the two common-ground decay directions.
Use the ordered orthonormal basis
\begin{equation*}
 (P_1,P_2,X_{01},Y_{01},X_{02},Y_{02},X_{12},Y_{12})
\end{equation*}
of the transverse quotient $\X_{P_0}$ and the metric $M=I_8$.
\begin{enumerate}
\item Without dephasing ($\lambda_{\rm ao}=0$), the reference transverse
monodromy is
\begin{equation}
 S_{\rm n}=\operatorname{diag}\!\left(
 \frac{16}{25},\frac{16}{25},
 \frac45,\frac45,\frac45,\frac45,
 \frac{16}{25},\frac{16}{25}\right),
\label{eq:common-ground-noiseless-monodromy}
\end{equation}
and its contraction factor is $q_{\rm n}:=\|S_{\rm n}\|_M=4/5$.
\item With the known dephasing rate
$\lambda_{\rm ao}=\lambda=2\log(4/3)$ throughout all four segments, the
reference transverse monodromy is
\begin{equation}
 S_{\rm ao}=\operatorname{diag}\!\left(
 \frac{16}{25},\frac{16}{25},
 \frac{81}{320},\frac{81}{320},\frac{81}{320},\frac{81}{320},
 \frac{16}{25},\frac{16}{25}\right),
\label{eq:common-ground-always-on-monodromy}
\end{equation}
and its contraction factor is $q_{\rm ao}:=\|S_{\rm ao}\|_M=16/25$.
\end{enumerate}
\end{proposition}

\begin{proof}
The controller returns to its cold-boundary state along the signature
sequence in Eq.~\eqref{eq:common-ground-signatures}.
The decay envelopes commute with themselves along each hold, and the pulse
Hamiltonian commutes with itself along each swap path, so the integrated
endpoints follow by direct evaluation of Eqs.~\eqref{eq:damping-ramp} and
\eqref{eq:swap-ramp}.  After the swap, which fixes $P_0$ and interchanges the two excited
indices, the second raw hold is executed at the swapped controller vertex,
giving the other decay direction without target motion.

For a decay exposure $\alpha$, the diagonal quotient action of
$e^{\alpha A}$ from Eq.~\eqref{eq:common-ground-generators} on
the displayed transverse basis serves both the reference-word calculation
and the gain-error enclosure:
\begin{equation}
 \operatorname{diag}\left(e^{-\alpha},1,
 e^{-\alpha/2},e^{-\alpha/2},1,1,
 e^{-\alpha/2},e^{-\alpha/2}\right).
\label{eq:single-common-ground-decay}
\end{equation}
The excited-state swap exchanges $P_1$ with $P_2$ and the $01$ and $02$
coherence pairs. On the $12$ pair it fixes $X_{12}$ and negates $Y_{12}$.
Conjugating the second raw decay by the swap supplies the missing
common-ground branch.  With $\alpha=2\log(5/4)$, one has
$e^{-\alpha}=16/25$ and $e^{-\alpha/2}=4/5$, so multiplication gives
Eq.~\eqref{eq:common-ground-noiseless-monodromy}.

The dephasing $\Diss{P_0}$ commutes with the excited-subspace swap and, at rate
$2\log(4/3)$, leaves populations and the $12$ coherence unchanged while giving
each ground--excited coherence the survival factor $3/4$ per unit-duration
segment.  Across four segments this contributes $(3/4)^4=81/256$, which,
multiplied by the damping factor $4/5$, gives $81/320$ in
Eq.~\eqref{eq:common-ground-always-on-monodromy}.

The largest absolute diagonal entry of each reference monodromy gives its
stated induced norm.
\end{proof}

\subsection{Reference monodromy and static Lyapunov input}

The diagonal actions identify the static Lyapunov input for each branch of
the robust instance.  In the transverse basis of
Proposition~\ref{prop:common-ground-word}, set
\begin{equation}
 G_{\rm cg}=\operatorname{diag}
 \left(-1,-1,-\tfrac12,-\tfrac12,-\tfrac12,-\tfrac12,-1,-1\right).
\label{eq:common-ground-static-generator}
\end{equation}
For the no-dephasing branch take
\begin{equation}
 A_L^{\rm n}=\alpha G_{\rm cg},\qquad M=I_8,\qquad
 \nu_{\rm n}=\frac\alpha2.
\label{eq:common-ground-static-no-dephasing}
\end{equation}
For the always-on branch define
\begin{equation}
 Z_q=\operatorname{diag}
 \left(0,0,-\tfrac12,-\tfrac12,-\tfrac12,-\tfrac12,0,0\right),\qquad A_L^{\rm ao}=\alpha G_{\rm cg}+4\lambda Z_q,\qquad
 \nu_{\rm ao}=\alpha.
\label{eq:common-ground-static-always-on}
\end{equation}
Since $(4/3)^4>5/4$, one has $4\lambda>\alpha$.  Thus, for
$b\in\{\mathrm n,\mathrm{ao}\}$,
\begin{equation}
 (A_L^b)^\dagger+A_L^b\preceq-2\nu_bI_8,\qquad
 e^{-\nu_{\rm n}}=\frac45,\qquad e^{-\nu_{\rm ao}}=\frac{16}{25}.
\label{eq:common-ground-static-receipts}
\end{equation}
The transverse monodromies of the two reference words are exactly
$e^{A_L^{\rm n}}$ and $e^{A_L^{\rm ao}}$, respectively, so their nominal
synthesis error is zero.

\subsection{Uniform gain enclosure}

The complete generators commute along each hold, so the uniform gain error
can be obtained from the diagonal endpoint action. Set the gain radius to
$\rho=1/20$ and
\begin{equation}
 x=e^{\rho\alpha}=(25/16)^{1/20}<\frac{131}{128}.
\label{eq:qutrit-root-enclosure}
\end{equation}
The last inequality is the integer comparison
$16\cdot131^{20}>25\cdot128^{20}$.  For either damping hold, the largest
induced-$2$ endpoint difference over $m\in[1-\rho,1+\rho]$ occurs in a
population coordinate and satisfies
\begin{equation}
 \delta_D=\frac{16}{25}(x-1)<\frac3{200}.
\label{eq:qutrit-hold-error}
\end{equation}
Here the physical and reference clocks coincide and their normalized complete
generators differ by $(m-1)a(s)A^{\rm tr}$, so substitution in
Eq.~\eqref{eq:primitive-duhamel-identity} evaluates to the displayed diagonal
endpoint difference because these generators commute along the path.  The
coherence difference is smaller, so the two holds use $\delta_1=\delta_3=\delta_D$
in Eq.~\eqref{eq:primitive-block-bounds}. Both directed pulses are exact, so
$\delta_2=\delta_4=0$. Every reference and actual transverse block,
including the known dephasing, is a Hilbert--Schmidt contraction, so
$c_k=r_k=1$ and every partial factor has norm at most one.  The fully
amplified telescope therefore gives
\begin{equation}
 E_{\rm hw}=\delta_1+\delta_3<\frac3{100},
 \qquad B_{\rm rob}=1.
\label{eq:qutrit-amplified-error}
\end{equation}
At $h=1$, combining $E_{\rm syn}=0$ with
Eqs.~\eqref{eq:common-ground-noiseless-monodromy},
\eqref{eq:common-ground-always-on-monodromy}, and
\eqref{eq:qutrit-amplified-error} yields
\begin{equation}
 \frac45+\frac3{100}=\frac{83}{100},\qquad
 \frac{16}{25}+\frac3{100}=\frac{67}{100},
\label{eq:qutrit-robust-arithmetic}
\end{equation}
giving the contraction margins in
Proposition~\ref{prop:common-ground-robust}.  Setting $\rho=0$ makes
$\delta_D=0$ and recovers the two reference contractions.

\subsection{Reference choice and complete prefixes}
To compare the always-on word with the no-dephasing reference through its
induced-norm error, define the no-dephasing reference norm
$q_{\rm id}:=\|S_{\rm n}\|_2$ and the induced-norm comparison error
$E_{\rm word}:=\|S_{\rm ao}-S_{\rm n}\|_2$.  Direct evaluation gives
\begin{equation}
 q_{\rm id}=\frac45,\qquad E_{\rm word}=\frac{35}{64},\qquad
 q_{\rm id}+E_{\rm word}=\frac{431}{320}>1,
\label{eq:failed-comparison-successful-actual}
\end{equation}
while the norm of the always-on reference monodromy is $16/25$ and the robust bound is
$67/100$.  Thus the comparison with the no-dephasing reference is inconclusive
about stability.

Replacing a directly callable $C$ hold by the swap--$A$--inverse-swap
implementation changes the intermediate paths, so the prefix bound must be
computed for that implementation. Let
$B_{\rm word}$ denote the sharp all-prefix induced norm of the physical word.
For the
Hilbert--Schmidt metric $M=I_8$ used in the resource calculation, every partial
swap is orthogonal and every damping or dephasing prefix is a contraction.
The swap can change the transverse direction without increasing its
Hilbert--Schmidt norm. Midway through the reference word's first swap,
\begin{equation}
 \frac45Y_{12}\longmapsto
 \frac{4}{5\sqrt2}(P_2-P_1),\qquad
 \left\|\frac{4}{5\sqrt2}(P_2-P_1)\right\|_{\rm HS}^2=\frac{16}{25}.
\label{eq:common-ground-prefix-witness}
\end{equation}
The uniform gain enclosure supplies $B_{\rm rob}=1$, while the
reference-prefix maps and trace-distance conversion below give the other
constants collected here:
\begin{equation}
 B_{\rm word}=B_{\rm rob}=1,\qquad K_{\rm in}=2,\qquad
 K_{\rm out}=\frac32,\qquad K_{\rm eff}=3.
\label{eq:common-ground-prefix-data}
\end{equation}
For an explicit verification of the reference-prefix bound, put
$\ell=\log(5/4)$ and, for local segment time $0\le\tau\le1$, define
\begin{equation}
 \beta(\tau)=
 \begin{cases}
  4\ell\tau^2,&0\le\tau\le1/2,\\
  2\ell-4\ell(1-\tau)^2,&1/2\le\tau\le1,
 \end{cases}\qquad \vartheta(\tau)=
 \begin{cases}
  \pi\tau^2,&0\le\tau\le1/2,\\
  \pi/2-\pi(1-\tau)^2,&1/2\le\tau\le1,
 \end{cases}\qquad d_\tau=e^{-\beta(\tau)/2}.
\label{eq:common-ground-partial-parameters}
\end{equation}
Thus $\chi:=d_1=4/5$.  The partial raw-decay map has quotient matrix
\begin{equation}
 D_\tau=\operatorname{diag}
 (d_\tau^2,1,d_\tau,d_\tau,1,1,d_\tau,d_\tau).
\label{eq:common-ground-partial-decay}
\end{equation}
Let $R_\theta$ be the quotient action of conjugation by
$U_\theta=e^{i\theta S_{12}}$, and write $c=\cos\theta$ and $s=\sin\theta$.
Its action on the population and excited-coherence block is
\begin{equation}
\begin{aligned}
 &R_\theta(P_1)=c^2P_1+s^2P_2+\sqrt2csY_{12},\qquad R_\theta(P_2)=s^2P_1+c^2P_2-\sqrt2csY_{12},\\
 &R_\theta(X_{12})=X_{12},\qquad R_\theta(Y_{12})=\frac{\sin(2\theta)}{\sqrt2}(P_2-P_1)
                   +\cos(2\theta)Y_{12},
\end{aligned}
\label{eq:common-ground-partial-swap-population}
\end{equation}
and on the ground-coherence block it is
\begin{equation}
\begin{aligned}
 &R_\theta(X_{01})=c^2X_{01}-csY_{01}+s^2X_{02}+csY_{02},\qquad R_\theta(Y_{01})=csX_{01}+c^2Y_{01}-csX_{02}+s^2Y_{02},\\
 &R_\theta(X_{02})=s^2X_{01}+csY_{01}+c^2X_{02}-csY_{02},\qquad R_\theta(Y_{02})=-csX_{01}+s^2Y_{01}+csX_{02}+c^2Y_{02}.
\end{aligned}
\label{eq:common-ground-partial-swap-coherences}
\end{equation}
In the normalized Hilbert--Schmidt basis, these equations give an orthogonal
matrix.  With $J=R_{\pi/2}=R_{-\pi/2}$, the four noiseless
prefix families are
\begin{equation}
 F_1(\tau)=D_\tau,\qquad F_2(\tau)=R_{\vartheta(\tau)}D_1,\qquad
 F_3(\tau)=D_\tau JD_1,\qquad F_4(\tau)=R_{-\vartheta(\tau)}D_1JD_1.
\label{eq:common-ground-prefix-families}
\end{equation}
The return edge has a separate controller label and follows the opposite
unitary path. Direct block multiplication gives
\begin{equation}
 \|F_1(\tau)\|_{2\to2}=\|F_2(\tau)\|_{2\to2}=1,\qquad
 \|F_3(\tau)\|_{2\to2}=d_\tau,\qquad \|F_4(\tau)\|_{2\to2}=\chi.
\label{eq:common-ground-prefix-norms}
\end{equation}
For instance, the population block of $F_3$ is
$\begin{psmallmatrix}0&d_\tau^2\\\chi^2&0\end{psmallmatrix}$. Each paired
ground-coherence block is
$\begin{psmallmatrix}0&d_\tau\\\chi&0\end{psmallmatrix}$. The $12$ block has
magnitude $d_\tau\chi$.  Orthogonality of the leftmost factor in $F_4$ reduces
its norm to that of $D_1JD_1$. Setting $\theta=\pi/4$ in
Eq.~\eqref{eq:common-ground-partial-swap-population} gives
Eq.~\eqref{eq:common-ground-prefix-witness}. Equations
\eqref{eq:common-ground-prefix-families}--\eqref{eq:common-ground-prefix-norms}
therefore prove the sharp value $B_{\rm word}=1$ for $M=I_8$.

For the always-on model, the accumulated dephasing over physical time $t$ is
\begin{equation}
 Z_t=\operatorname{diag}
 \left(1,1,(3/4)^t,(3/4)^t,(3/4)^t,(3/4)^t,1,1\right).
\label{eq:common-ground-dephasing-prefix}
\end{equation}
Since it commutes with $D_\tau$ and $R_\theta$, the prefix during segment $j$ at local time
$\tau$ is
$\widetilde F_j(\tau)=Z_{j-1+\tau}F_j(\tau)$, for which another block calculation
gives
\begin{equation}
 \|\widetilde F_1\|_{2\to2}=\|\widetilde F_2\|_{2\to2}=1,\qquad
 \|\widetilde F_3\|_{2\to2}=d_\tau^2,\qquad
 \|\widetilde F_4\|_{2\to2}=\chi^2.
\label{eq:common-ground-dephasing-prefix-norms}
\end{equation}
At $\tau=1$, $F_4(1)=S_{\rm n}$ and
$\widetilde F_4(1)=S_{\rm ao}$, so these families also recover both complete
monodromies in Proposition~\ref{prop:common-ground-word}.

The trace-distance conversion constants follow from orthogonal quotient
reconstruction. If $R_\rho=Q_\rho$ denotes the
Hilbert--Schmidt-orthogonal transverse reconstruction and
$\Delta=\rho-P_0=R_\rho-\tr(R_\rho)P_0$, then
\[
 \|R_\rho\|_{\rm HS}\le\|\Delta\|_1=2\dtr(\rho,P_0),\qquad
 \dtr(\rho,P_0)\le\frac32\|R_\rho\|_{\rm HS}.
\]
The first inequality uses nonexpansiveness of orthogonal projection. The
second follows from
$\|\Delta\|_{\rm HS}\le\sqrt3\|R_\rho\|_{\rm HS}$ and
$\|\Delta\|_1\le\sqrt3\|\Delta\|_{\rm HS}$.
For $M=I_8$, the choices $K_{\rm in}=2$ and $K_{\rm out}=3/2$ give equality
in the sufficient matrix inequalities $K_{\rm in}^2I-4M\succeq0$ and
$K_{\rm out}^2M-\tfrac94I\succeq0$.
Together with $B_{\rm word}=1$ these give $K_{\rm eff}=3$, proving all
constants in Eq.~\eqref{eq:common-ground-prefix-data}.

These Hilbert--Schmidt constants enter the resource account. The weighted
metric in the following comparison retains the physical word but changes
its sharp prefix bound.

\begin{proposition}[Prefix amplification]
\label{prop:prefix-bound-noninheritance}
In the transverse basis of Proposition~\ref{prop:common-ground-word}, let
\begin{equation*}
 M_\star=\operatorname{diag}(4,1,1,1,1,1,1,1).
\end{equation*}
This is a strict Lyapunov
metric for $A+C$. With exposure $\alpha=2\log(5/4)$ in each damping hold,
the direct two-label cycle generated by $A$ and $C$ has sharp all-prefix bound
one in $M_\star$,
whereas the physical word in Eq.~\eqref{eq:common-ground-word} has sharp
all-prefix bound two.  The latter value remains two in the always-on dephasing
model.
\end{proposition}

\begin{proof}
The partial decay, swap, and dephasing actions used here are displayed in
Eqs.~\eqref{eq:common-ground-partial-decay},
\eqref{eq:common-ground-partial-swap-population}, and
\eqref{eq:common-ground-dephasing-prefix}.

With $M_\star$, the quotient matrix $G_{\rm cg}$ in
Eq.~\eqref{eq:common-ground-static-generator} satisfies
\[
 G_{\rm cg}^{\mathsf T}M_\star+M_\star G_{\rm cg}
 =\operatorname{diag}\bigl(-8,-2,-1,-1,-1,-1,-2,-2\bigr)\prec0.
\]
Partial $A$ and $C$ decays are diagonal contractions in $M_\star$, so the
direct cycle has sharp prefix bound one.
Every quotient prefix of the physical word, being a product of Hilbert--Schmidt
contractions, is a Hilbert--Schmidt contraction and, since the
condition number of $M_\star$ is four, has induced
$M_\star$-norm at most two.  The prefix ending at the forward swap sends the unchanged input $P_2$
to $P_1$, and
\[
 \frac{\|P_1\|_{M_\star}}{\|P_2\|_{M_\star}}=2.
\]
This proves sharpness.  The population witness is untouched by the always-on
dephasing, which itself is a Hilbert--Schmidt contraction.
\end{proof}

\subsection{Resource account}
Table~\ref{tab:common-ground-costs} uses the constants in
Eq.~\eqref{eq:common-ground-prefix-data} for the closed-cycle resource account.
A raw call has scalar cost one and each pulse/transition edge has scalar cost
two, while per-switch wall-time and scalar costs are zero.  Pulse and transition counts
refer to the same two directed edges and are not added a second time to the
aggregate primitive count.  Initialization, finalization, and the closed seam
are identities with zero cost.  The fixed background dephasing
has no additional call or scalar charge in the declared dephasing presentation.

The signature sequence in Eq.~\eqref{eq:common-ground-signatures} has five
switches. To evaluate Eq.~\eqref{eq:cycle-count-branches}, we first apply the
joint worst-case gain and contraction envelopes and then take the integer
ceilings. With $K_{\rm eff}=3$, this gives the robust counts in the first row
below. At zero uncertainty the contractions are
$4/5$ and $16/25$, giving the certified minimum counts in the second row:
\begin{equation}
 \begin{aligned}
 &\left\lceil\frac{\log(3000)}{-\log(83/100)}\right\rceil=43,\qquad
 &&\left\lceil\frac{\log(3000)}{-\log(67/100)}\right\rceil=20,\\
 &\left\lceil\frac{\log(3000)}{-\log(4/5)}\right\rceil=36,\qquad
 &&\left\lceil\frac{\log(3000)}{-\log(16/25)}\right\rceil=18.
 \end{aligned}
\label{eq:common-ground-cycle-arithmetic}
\end{equation}
The integer immediately below each displayed value fails the corresponding
scalar inequality. Multiplying each per-cycle count, duration, scalar cost,
or control action by the robust cycle count gives its total in
Table~\ref{tab:common-ground-costs}. The gain-box maximum raw exposure
per cycle is
$2(21/20)\alpha=(21/10)\alpha$, hence the totals are
$(903/10)\alpha$ and $42\alpha$.  Known dephasing acts for four unit intervals
per cycle, giving total exposure $80\lambda$ in the twenty-cycle branch.

\begin{table}[H]
\caption{Worst-case resources for the common-ground word at $d_0=1$ and
$\eps=10^{-3}$.  The no-dephasing and dephasing columns use
$\bar q=83/100$ and $67/100$. Logical, raw, pulse, and transition calls are
overlapping componentwise accounts, whereas aggregate calls count each primitive once.}
\label{tab:common-ground-costs}
\begin{ruledtabular}
\begin{tabular}{lrrr}
quantity & worst case per cycle & no dephasing & dephasing\\
\hline
cycles & 1 & 43 & 20\\
wall time & 4 & 172 & 80\\
aggregate primitive calls & 4 & 172 & 80\\
logical orbit calls & 2 & 86 & 40\\
raw dissipator calls & 2 & 86 & 40\\
pulse calls & 2 & 86 & 40\\
transition calls & 2 & 86 & 40\\
reset calls & 0 & 0 & 0\\
dead calls & 0 & 0 & 0\\
cooldown calls & 0 & 0 & 0\\
physical-signature switches & 5 & 215 & 100\\
scalar cost & 6 & 258 & 120\\
integrated control action & $\pi$ & $43\pi$ & $20\pi$\\
raw dissipative exposure & $\le(21/10)\alpha$ & $\le(903/10)\alpha$ & $\le42\alpha$\\
always-on exposure & $0$ or $4\lambda$ & $0$ & $80\lambda$
\end{tabular}
\end{ruledtabular}
\end{table}

\paragraph{A gain box containing zero.}
\label{prop:common-ground-radius-one-obstruction}
No uniform strict contraction holds for the same word on
$(m_1,m_3)\in[0,2]^2$.  At $m_1=m_3=0$, the word without dephasing is the
identity channel.  With $\mathcal Z$ present, excited populations still contain
a unit transverse mode.

\section{Complete-path certification for the four-level network}
\label{app:four-level-certification}

In the shadow model, which omits pumping out of the target, we first propagate
the survival effect through complete holds, pulses, and dead intervals. Uniform
cycle bounds are then transferred to the actual paths and combined with
prefix bounds to certify preparation and continued operation.

\subsection{Adjoint dynamics and certified enclosures}

To propagate the target-complement observable backward through the commands,
use the four-dimensional real space of Hermitian operators supported on the
excited subspace with the form
\begin{equation}
 E=e_1P_1+e_2P_2+e_3P_3
       +if|2\rangle\langle3|-if|3\rangle\langle2|.
\label{eq:four-effect-coordinates}
\end{equation}
Here $e_1,e_2,e_3,f$ are real coordinates. Because the operator
$Q=P_1+P_2+P_3$ belongs to this space and every shadow adjoint command
preserves this four-dimensional space, starting from $Q$ reduces the complete
survival-effect calculation to these four coordinates.
For a hold of actual duration $\tau$, let $x=g\tau$, $d=\kappa\tau$,
and $\xi=e^{-x}$.  Integrating the generator equations
$\dot e_1=-ge_1$, $\dot e_2=ge_1-ge_2$, $\dot e_3=0$, and
$\dot f=-(g/2+\kappa)f$ gives
\begin{equation}
 (e_1,e_2,e_3,f)\mapsto
 (\xi e_1,x\xi e_1+\xi e_2,e_3,e^{-x/2-d}f).
\label{eq:four-hold-adjoint}
\end{equation}
In particular, the transfer from level two to level one is retained.

For a pulse with gain $a$, let $\alpha=\pm4a\tau$ be its actual signed angle,
$d=\kappa\tau$, $z=(e_2-e_3)/2$, and $m=(e_2+e_3)/2$.
The pair $(z,f)$ obeys $\dot z=2uf$ and
$\dot f=-2uz-\kappa f$, where $u$ is the signed physical amplitude.
Set $r=\sqrt{4\alpha^2-d^2/4}$ and $s=\sin(r)/r$. On the present domain,
$r>0$, and the pulse solution is
\begin{align}
 &z'=e^{-d/2}\left[(\cos r+ds/2)z+2\alpha sf\right],\label{eq:four-pulse-adjoint-first}\\
 &f'=e^{-d/2}\left[-2\alpha sz+(\cos r-ds/2)f\right],\\
 &e_1'=e_1,\qquad e_2'=m+z',\qquad e_3'=m-z'.
\label{eq:four-pulse-adjoint}
\end{align}
The simultaneous dephasing term is part of this solution.  A dead interval
multiplies $f$ by $e^{-1/2000}$ and leaves the diagonal coordinates fixed.
To obtain $\mathcal S^*(Q)$, apply these adjoints in reverse chronological
order to $(1,1,1,0)$.

The interval calculation can be reproduced with integer endpoints at scale
$2^{80}$.  Addition, multiplication, division, and square root are rounded
outwards using integer arithmetic.  For the exponential, range reduction
to $|x|\le1$ permits degree-32 Taylor expansion with remainder at most
$3/33!$, followed by squaring.  For sine and cosine at scalar arguments
$|x|\le8$, Taylor degrees 81 and 80 have remainders at most
$8^{82}/82!$ and $8^{81}/81!$.  An interval argument is treated by the
addition formulas around its dyadic midpoint, using
$\sin\delta\in[-r_0,r_0]$ and
$\cos\delta\in[1-r_0^2/2,1]$ for $|\delta|\le r_0\le1$.
All arguments in the displayed hardware box obey these range restrictions,
so the interval procedure supplies a finite calculation with retained truncation and rounding
errors.  Numerical matrix exponentials provide a separate implementation
check of the four-coordinate formulas.

\subsection{Shadow cycle certificates}

A hold retains the dark state $P_3$, whereas a pulse retains $P_1$.
For the complete word
\begin{equation}
 A(3)\longrightarrow P(157/100)\longrightarrow A(5),
\label{eq:four-level-simple-word}
\end{equation}

the shadow survival effect and actual cycle duration satisfy
\begin{equation}
 \mathcal S^*(Q)\preceq\frac{2}{25}Q,\qquad
 T_{\rm cyc}\le T:=\frac{174207}{20000}=8.71035
\label{eq:four-level-simple-certificate}
\end{equation}
throughout the gain and timing box. For Eq.~\eqref{eq:four-level-simple-word},
outward interval evaluation gives the rationally rounded coordinate bounds
\begin{equation}
 0\le e_1\le5.83\times10^{-4},\quad
 e_2\le0.06354,\quad e_3\le0.058531,\quad |f|\le0.015861.
\label{eq:four-simple-effect-bounds}
\end{equation}
All decimal endpoints in this display are terminating rationals rounded
outwards. The coordinate enclosure bounds the maximal eigenvalue of the
$2\times2$ excited block, which is nondecreasing in both diagonal entries and
in $|f|$, so substituting their upper endpoints gives a value below
$0.0771<2/25$. Together with the remaining eigenvalue $e_1$, this gives
$\mathcal S^*(Q)\preceq(2/25)Q$.
Taking the trace of this operator inequality against any density operator
proves the survival bound.

For longer words, cancellation can be retained by expanding the pulse map
before interval evaluation.  If $A=e^{-d/2}(\cos r+ds/2)$,
$C=2\alpha se^{-d/2}$, and $D=e^{-d/2}(\cos r-ds/2)$, the equivalent formulas are
\begin{align}
 &e_2'=(1+A)e_2/2+(1-A)e_3/2+Cf,\\
 &e_3'=(1-A)e_2/2+(1+A)e_3/2-Cf,\\
 &f'=-Ce_2/2+Ce_3/2+Df.
\end{align}
The coordinatewise intersection of this enclosure with the enclosure obtained
from Eqs.~\eqref{eq:four-pulse-adjoint-first}--\eqref{eq:four-pulse-adjoint} still contains the exact effect.  With this
evaluation, Eq.~\eqref{eq:four-level-five-word} has
\begin{equation}
 e_1<5.543\times10^{-6},\quad e_2<0.007295,\quad e_3<0.005336,
 \quad |f|<0.002777.
\end{equation}
The spectral upper bound obtained from these coordinates is below
$0.0093<47/5000$.

\subsection{Pumping and all-time preparation}

Write $\lambda_{\max}=10^{-5}$ for the pumping-rate upper bound in
Table~\ref{tab:four-level-guarantee}. Every shadow prefix has nonincreasing
survival probability. For complete
actual paths, including pumping during each hold, pulse, and dead interval,
the estimate $\|\lambda\Diss{|1\rangle\langle0|}\|_\diamond\le2\lambda$
and CPTP Duhamel propagation bound the trace norm of the
actual--shadow state difference by $2\lambda_{\max}T$ at every prefix
when $T$ bounds the actual cycle duration. Since this difference is traceless Hermitian, a
projector expectation changes by at most $\lambda_{\max}T$, giving the
survival-probability bound $\eta=10^{-5}T$ for the pumping contribution to a
whole cycle or prefix.
For the three-command word, the cycle recursion is
$p_{n+1}\le(2/25)p_n+\eta$.
Within a cycle, $\dot p=-g\rho_{11}+\lambda\rho_{00}\le\lambda_{\max}$
on a hold and $\dot p=\lambda\rho_{00}\le\lambda_{\max}$ otherwise,
so every actual prefix obeys $p(t)\le p_n+\eta$.

Let $t_2$ denote the end of the second actual cycle. To establish
\begin{equation}
 d_{\rm tr}(\rho(t),P_0)\le0.1
 \quad\text{for all }t\ge t_2,\qquad t_2\le17.4207,
\label{eq:four-level-preparation}
\end{equation}
for every initial density operator under continued execution, apply
Proposition~\ref{prop:finite-search-test} with $B=1$, $b=\eta$, $Y=1$, and
$\psi(v)=\sqrt v$ from Eq.~\eqref{eq:survival-distance}. For $q=2/25$,
$T=174207/20000$, and $\eta=T/10^5$, the limiting value $\eta/(1-q)$ is
smaller than the two-cycle boundary bound, and
\begin{equation}
 q^2+(1-q^2)\frac{\eta}{1-q}+\eta<0.0066<0.01.
\label{eq:four-two-cycle-tail}
\end{equation}
The test at $N=2$ and the duration bound $t_2\le2T$ establish
Eq.~\eqref{eq:four-level-preparation}.

For the five-command word, the duration bound
is $T=143107/10000$. With $\eta=T/10^5$, after one complete cycle every
subsequent prefix obeys $p(t)\le47/5000+2\eta<0.0097<0.01$.
The limiting bound from the scalar recursion is smaller than its first-cycle
boundary bound. This proves the stated one-cycle preparation guarantee for this word.

\subsection{Numerically proposed five-command word}
\label{app:optimized-word}

The numerical proposal in Table~\ref{tab:four-level-comparison} is
\begin{equation}
 A(3.1725)\longrightarrow P(1.5702)\longrightarrow A(5.5998)
 \longrightarrow P(1.5702)\longrightarrow A(3.9475).
\label{eq:four-level-optimized-word}
\end{equation}
For these exact terminating rational parameters, the same outward
calculation gives $\mathcal S^*(Q)\preceq qQ$ with $q=9.7191\times10^{-3}$ and
the cycle-duration upper bound $T=14.024998$ for this proposal. The
continuous-parameter minimum remains undetermined.

With $\eta=10^{-5}T$, after one complete preparation cycle every subsequent
prefix obeys
\begin{equation}
 p(t)\le q+2\eta=10^{-2}-4.0004\times10^{-7}<10^{-2}.
\label{eq:four-level-optimized-tail}
\end{equation}
The limiting boundary radius $\eta/(1-q)$ is smaller than $q+\eta$.
Thus repetition of the proposed word maintains trace distance at most $0.1$
after the first cycle.

\subsection{Preparation followed by the matching periodic tail}
\label{app:periodic-continuation}

For the periodic comparison, the preparation prefix may end halfway through
an $AP$ period, so its continuation certificate must use the controller's
actual scheduler phase at that boundary. Let $\beta_{\rm pre}$ bound the shadow
survival effect of the entire preparation prefix, and let $T_{\rm pre}$ bound
its physical duration, including each dead interval. The preceding Duhamel
estimate gives, for every initial state,
\begin{equation}
 p_{\rm pre}\le \beta_{\rm pre}+\lambda_{\max}T_{\rm pre}.
\label{eq:periodic-preparation-bound}
\end{equation}
Choose a two-period tail, either $APAP$ or $PAPA$, that begins at the resulting
phase. A uniform shadow factor $\beta_{\rm tail}<1$ and duration bound
$T_{\rm tail}$ give
\begin{equation}
 \eta_{\rm tail}=\lambda_{\max}T_{\rm tail},\qquad
 r_{\rm tail}=\frac{\eta_{\rm tail}}{1-\beta_{\rm tail}},\qquad
 p_{n+1}\le\beta_{\rm tail}p_n+\eta_{\rm tail}.
\end{equation}
Hence each later tail boundary has survival probability at most
$\max\{p_{\rm pre},r_{\rm tail}\}$. The within-tail differential inequality
$\dot p\le\lambda_{\max}$ then proves the sufficient all-time criterion
\begin{equation}
 \max\!\left\{\beta_{\rm pre}+\lambda_{\max}T_{\rm pre},
              r_{\rm tail}\right\}+\eta_{\rm tail}\le\eps^2.
\label{eq:periodic-continuation-test}
\end{equation}
All bounds are uniform over the declared joint error set and apply to
successive tails even when errors change or are correlated across tails.

For the reported periodic-grid winner, the common hold is $3.4$ and the pulse
angle is $1.55$. The preparation word $APAPAPA$ has
$T_{\rm pre}=15.40775$, and its next command is $P$, so the continuing blocks are
$PAPA$ with $T_{\rm tail}=7.9265$. The same outward primitive calculation
gives
\begin{equation}
 \beta_{\rm pre}<0.009664,\qquad \beta_{\rm tail}<0.183943.
\label{eq:periodic-rounded-factors}
\end{equation}
These displayed terminating decimals are rational upper bounds. The
preparation term exceeds $r_{\rm tail}$, and
the left-hand side of Eq.~\eqref{eq:periodic-continuation-test} is bounded by
$0.0099<0.01$. Equation~\eqref{eq:survival-distance} proves the claimed
trace-distance tolerance throughout subsequent operation.

For each finite word in the periodic search, the checker also considers every
contiguous block partition. If the blocks have uniform shadow inequalities
$\mathcal S_i^*(Q)\preceq\beta_iQ$, positivity of their adjoints gives
\begin{equation}
 (\mathcal S_m\cdots\mathcal S_1)^*(Q)
 \preceq\left(\prod_{i=1}^m\beta_i\right)Q.
\end{equation}
The minimum of these products and the directly enclosed whole-word factor
remains valid, including under correlated admitted errors. The companion
archive retains the exact rational factors, enumeration domains, and matched
comparisons.

\subsection{Finite-library refinement}
\label{app:search-refinement}
A discrete search branch fixes the word template and any integer choices,
such as $N$, while parameter boxes and arithmetic accuracy remain to be
refined. On each branch, fix nested finite subdivisions of its parameter
boxes whose maximal diameter tends to zero, and index increasing arithmetic
accuracy by $m\in\mathbb N$.  A refinement is \emph{fair} if, until a
certificate is found, it eventually evaluates every triple consisting of a
represented branch, a box at a finite subdivision depth, and a finite
accuracy level, unless that box or one of its ancestors has been discarded
by a valid exclusion as specified below.
Each evaluation encloses all required inequalities on that box.  A box and
its descendants may be discarded only when a valid enclosure excludes the required
certificate throughout the box.

\begin{proposition}[Scope of refinement]
\label{prop:search-refinement}
On each discrete branch, suppose exact identities have been parameterized and
the remaining conditions are finitely many continuous inequalities
$f_j(x)\le0$, $1\le j\le J$, on finitely many compact coordinate boxes.
Suppose computable interval enclosures converge uniformly under subdivision
and increasing arithmetic accuracy.  A fair refinement finds a certificate
in finite time if a represented point $x_*$ on some branch satisfies
$f_j(x_*)\le-\sigma$ for every $1\le j\le J$ and some $\sigma>0$.
\end{proposition}

\begin{proof}
If $f_j(x_*)\le-\sigma$ for every $j$ and $\sigma>0$, continuity gives a
relative neighborhood on which $f_j\le-\sigma/2$.  A sufficiently fine box
containing $x_*$ lies in that neighborhood.  At sufficiently high accuracy,
the computed upper endpoint for each $f_j$ lies less than $\sigma/4$ above
its true supremum on that box, so all these upper bounds are negative.
Fair refinement eventually reaches such a box and
accuracy level, and a valid exclusion cannot discard a box containing the
strictly feasible point.
\end{proof}

For fixed $N$, the objective $NT(x)$ is continuous on a branch with a
continuous certified cycle-duration bound $T(x)$. In the four-level example,
the fixed timing box gives a bound affine in the nominal holds and angles
for each template. Once a certified execution-time upper bound $U$ is known
and each cycle lasts at least $T_->0$, a better certificate has
$N\le\lfloor U/T_-\rfloor$.
For a nonempty finite list of candidates with evaluated certificates,
exhaustive comparison selects the smallest certified execution-time upper
bound among the accepted candidates, when any are accepted. The refinement
proposition above instead guarantees discovery of a represented strictly
feasible certificate under its stated assumptions.

\section{Direct-label schedules and interstitial bounds}
\label{sec:direct}

When the calibrated atoms are callable directly, the active schedule
approximates the static evolution, while transitions and dead intervals
contribute at their actual positions in the word. Both estimates use the static Lyapunov metric.

\subsection{Active schedules and nominal error}

For the compatible direct realization from Definition~\ref{def:direct-realization},
let $A_j$ be its transverse atoms. At cycle parameter $h>0$, the first-order
exclusive schedule assigns the physical dwell durations and active product
\begin{equation}
 t_j(h)=\frac{\lambda_jh}{r_j},\qquad B_j=\lambda_jA_j,\qquad
 S_1(h)=e^{hB_m}\cdots e^{hB_1}.
\label{eq:first-order-word}
\end{equation}
Thus $t_j(h)\widetilde A_j=h\lambda_jA_j$ for every dwell.
For a permutation $\pi=(\pi_1,\ldots,\pi_m)$ of
$\{1,\ldots,m\}$, a cut-open symmetric schedule uses
$S_2(h)=e^{hB_{\pi_1}}$ for $m=1$. For $m\ge2$, it uses
\begin{equation}
 S_2(h)={}e^{hB_{\pi_1}/2}\cdots e^{hB_{\pi_{m-1}}/2}
 e^{hB_{\pi_m}} {}\times e^{hB_{\pi_{m-1}}/2}\cdots e^{hB_{\pi_1}/2}.
\label{eq:strang-word}
\end{equation}
Every coefficient is positive and each label receives total normalized
exposure $\lambda_jh$.  Because the cut at the cold boundary is part of the
presentation, half arcs are not merged across periods without an explicit
seam rule.

We compare the nominal active product with static decay in $\|\cdot\|_M$,
the induced norm of the static Lyapunov metric. For a chosen schedule order
$o\in\{1,2\}$, suppose certified upper bounds satisfy
\begin{equation}
 \|S_o(h)-e^{hA_L}\|_M\le \overline E_o(h),\qquad
 \|e^{hA_L}\|_M\le\overline d_\nu(h).
\label{eq:product-bounds-summary}
\end{equation}
For example, if nonnegative norm bounds $b_j,c_{kj}\ge0$ satisfy
$B_j^\dagger MB_j\preceq b_j^2M$ and
$[B_k,B_j]^\dagger M[B_k,B_j]\preceq c_{kj}^2M$, then with
$\beta=\sum_jb_j$ and $C_2=\sum_{j<k}c_{kj}$,
\begin{equation}
 \overline E_1(h)=\frac{h^2}{2}e^{h\beta}C_2,\qquad
 \overline E_2(h)=\frac{(h\beta)^3}{3}e^{h\beta}
\label{eq:product-error-bounds}
\end{equation}
are sufficient conservative bounds.  If the transverse atoms commute
pairwise, one may take $c_{kj}=0$, giving $C_2=0$ and
$\overline E_1(h)=0$. The actual symmetric product
 also equals $e^{hA_L}$, although the generic bound $\overline E_2(h)$ displayed
above need not vanish.  Related positive-coefficient and open-system
product estimates are well established
\cite{kliesch2011churchturing,iserles2024splitting,childs2017sparse}.

\subsection{Interstitial dynamics and cycle contraction}

Let $E_\ell$ and $D_\ell$ denote the induced transverse maps of the active and
interstitial channels, respectively, retained in chronological order.
If $S=E_p\cdots E_1$ is the active product and
$F=D_pE_p\cdots D_1E_1$ includes transition or dead maps, then for $p\ge2$
\begin{equation}
 F=R_{\rm int}S,\qquad R_{\rm int}=D_p(Q_{p-1}D_{p-1}Q_{p-1}^{-1})\cdots \quad\times(Q_1D_1Q_1^{-1}),
\label{eq:interstitial-factorization}
\end{equation}
where $Q_\ell=E_p\cdots E_{\ell+1}$ and $Q_p=I$. For $p=1$, $F=D_1S$ and
$R_{\rm int}=D_1$. The active factors are exponentials, so these suffixes
are invertible. Choose certified bounds
\[
 \|D_\ell-I\|_M\le\Delta_\ell,\qquad
 \|Q_\ell\|_M,\ \|Q_\ell^{-1}\|_M\le e^{hs_\ell}
 \quad(1\le\ell\le p).
\]
Bounding each interstitial perturbation after conjugation by its active
suffix gives the multiplicative enclosure
\begin{equation}
 \|R_{\rm int}\|_M\le
 \overline\kappa_{\rm int}(h)
 =\prod_{\ell=1}^p\left(1+e^{2hs_\ell}\Delta_\ell\right).
\label{eq:interstitial-bound}
\end{equation}
The active suffix and its inverse retain the amplification caused by the
interstitial map's position. Moving that map outside the word would discard
these conjugations.

Write $T_{\rm cyc}$ for the physical duration of one closed direct-label
cycle.

\begin{theorem}[Direct-label switching certificate]
\label{thm:direct-compiler}
Take a finite static certificate and its Lyapunov metric $M$.  Assume a
 target-compatible direct realization, either the ordered first-order schedule
 ($o=1$) or a cut-open symmetric schedule ($o=2$), and target-compatible CPTP
 paths for every interstitial
interval.  Suppose every dwell and calibration constraint holds and a
certified enclosure proves
\begin{equation}
 \overline q(h)=\overline\kappa_{\rm int}(h)
 \bigl(\overline d_\nu(h)+\overline E_o(h)\bigr)<1.
\label{eq:direct-qbar}
\end{equation}
Then the repeated direct-label cycle leaves $\F_P$ invariant at every physical
time and is GAS to that face. For a state trajectory $\rho(t)$ under this
repeated cycle, write $x(t):=z(\rho(t))$ for the transverse coordinates
of Sec.~\ref{sec:presentations}. At cycle boundaries,
\begin{equation}
 \|x(nT_{\rm cyc})\|_M\le\overline q(h)^n\|x(0)\|_M.
\label{eq:direct-cycle-contraction}
\end{equation}
If every induced transverse within-cycle prefix has $M$-norm at most $B$,
the same bound at a physical time $nT_{\rm cyc}+s$ acquires the factor $B$.
\end{theorem}

For a fixed finite realization, the ordered factorization bounds the cycle,
and Lyapunov contraction gives quotient decay at its boundaries
(Appendix~\ref{app:monodromy}). Each CPTP prefix preserves $\F_P$ and cannot
increase trace distance to it, so attraction holds at all times.
When an atom is realized by uncertain finite-duration
paths, Sec.~\ref{sec:physical-words} derives both the actual contraction and
the prefix bound from the complete primitive data.

\subsection{Dwell constraints and finite support}

If $\tau_j$ is the minimum permitted dwell for label $j$, the first-order
dwell lower endpoint is
\begin{equation}
 h_{\min,1}^{\rm dwell}=\max_j\frac{r_j\tau_j}{\lambda_j},
\label{eq:first-order-dwell}
\end{equation}
whereas a noncentral atom in Eq.~\eqref{eq:strang-word} obeys the stronger
factor-two requirement. A nonempty open interval of certified exposures
requires the maximum of the dwell and interstitial lower bounds to lie
strictly below the calibration and stability upper endpoint. A single
exposure can be certified by checking its constraints directly.
For positive minimum dwells, the timing
quantity $\min_j\lambda_j/(r_j\tau_j)$ is the reciprocal of the first-order dwell
lower endpoint in Eq.~\eqref{eq:first-order-dwell}. This quantity is positive for each fixed
finite realization, while a uniform positive lower bound over a family must be
established separately. The following fixed-support family illustrates the distinction.
\paragraph{Finite support does not control timing.}
For the common-ground family
\begin{equation}
 L_\delta=\Diss{|0\rangle\langle1|}
 +\delta\Diss{|0\rangle\langle2|},\qquad
 \delta=2^{-n},\quad n\in\mathbb N,
\label{eq:small-weight-family}
\end{equation}
the dimension and support size remain fixed while the two direct weights are
$(1,\delta)$.  With unit rates and a common minimum dwell $\tau$, the
first-order schedule must satisfy
\begin{equation}
 h_{\min}^{\rm dwell}=\tau/\delta.
\label{eq:small-weight-dwell}
\end{equation}
Thus finite Carath\'eodory support supplies neither a uniform positive lower
bound on weights nor a dimension-only dwell guarantee.

\section{A cascade ladder with a sharp presentation-relative law}
\label{subsec:ladder}

This example connects an exact finite factorization to a sharp resource bound
in a two-label presentation. Its periodic schedule is cut at a chosen
observation phase, so factor pieces and physical active arcs are counted
separately.

\subsection{Exact finite factorization}

Set
\begin{equation}
 A_\ell=\Diss{|0\rangle\langle1|},\qquad
 B_\ell=\Diss{|1\rangle\langle2|},\qquad
 L_\ell=A_\ell+B_\ell.
\label{eq:ladder-generators}
\end{equation}
Unlike the common-ground pair, these generators do not commute.  Their
commutator $C_\ell=[A_\ell,B_\ell]$ obeys
\begin{equation}
 C_\ell(X)=X_{22}(P_0-P_1),\qquad
 [A_\ell,C_\ell]=-C_\ell,\qquad [B_\ell,C_\ell]=C_\ell.
\label{eq:ladder-commutator}
\end{equation}

The generator-level Lyapunov input is recorded in
Proposition~\ref{prop:ladder-certificate}.

\begin{theorem}[Exact ladder factorization]
\label{thm:ladder-factorization}
For every $T>0$, define
\begin{equation}
 x(T)=\log\frac{T}{1-e^{-T}},\qquad
 y(T)=T-x(T)=\log\frac{e^T-1}{T}.
\label{eq:ladder-times}
\end{equation}
Then $0<x(T),y(T)<T$ and
\begin{equation}
 e^{T(A_\ell+B_\ell)}
 =e^{y(T)A_\ell}e^{TB_\ell}e^{x(T)A_\ell}.
\label{eq:ladder-factorization}
\end{equation}
No word in the two primitive semigroups with positive factor durations and
fewer than three factor pieces equals the left-hand side.
\end{theorem}

\begin{proof}
The ideal population coefficient from level two to level one is $Te^{-T}$,
matching the right-hand-side coefficient $e^{-y}(1-e^{-T})$ by
Eq.~\eqref{eq:ladder-times}.  The total exposure to each generator is $T$,
giving the same factors $e^{-T/2}$ on the $01$ and $02$ coherences and $e^{-T}$
on the $12$ coherences.  Conversely, the $01$ and $02$ coherences force
both total exposures to be $T$.  A two-factor word then gives either
$e^{-T}(1-e^{-T})$ or $1-e^{-T}$ for that population coefficient, and
\begin{equation*}
 e^{-T}(1-e^{-T})<Te^{-T}<1-e^{-T}.
\end{equation*}
This proves both claims on the full matrix-unit basis.
\end{proof}

Equation~\eqref{eq:ladder-factorization} has zero error at finite positive
durations, so this model has no positive dwell error floor.

\subsection{Periodic schedule and resource bounds}

The physical grammar for the resource statement has exclusive unit-amplitude
access to $A_\ell$ and $B_\ell$, minimum active-arc dwell $h>0$, and an
identity dead interval of duration $\sigma>0$ at each label change.
For switch counting, each complete label invocation includes its following
dead interval under the same physical signature. For $T\ge h$, one $A_\ell$ arc
and one $B_\ell$ arc, each of duration $T$, form a period
\begin{equation}
 \mathcal P(T)=2T+2\sigma.
\label{eq:ladder-period}
\end{equation}
At the observation phase $y(T)$ into the $A_\ell$ arc, cutting the period
gives the factorization in Eq.~\eqref{eq:ladder-factorization}
with $N_{\rm factor}=3$ pieces from two physical arcs. The dwell constraint
concerns the uncut $A_\ell$ arc of duration $T$. State injection and
observation occur only after the
periodic controller has been initialized, and the resource clock begins at
this strobe.  Initialization and pre-running are excluded from both sides of
the comparison.

The factorizing schedule has the cost vector
\begin{equation}
 (T_{\rm wall},X_A,X_B,N_{\rm arc},N_{\rm factor},N_{\rm sw},
   A_{\max},B_{\rm sw},T_{\rm pulse}) =(2T+2\sigma,T,T,2,3,2,1,\sigma^{-1},0).
\label{eq:ladder-cost-vector}
\end{equation}
Here $T_{\rm wall}$ is elapsed period time, and $X_A$ and $X_B$ are cumulative
generator exposures. The counts $N_{\rm arc}$, $N_{\rm factor}$, and
$N_{\rm sw}$ refer respectively to maximal active arcs on the period circle,
pieces after cutting at the observation strobe, and changes between complete
invocation signatures, including any change across the cyclic seam.
The remaining entries are the amplitude cap $A_{\max}$, the interface
transition-rate parameter $B_{\rm sw}=1/\sigma$ (not a Fourier bandwidth), and
the Hamiltonian-pulse time $T_{\rm pulse}$. There is no Hamiltonian pulse, so
Eq.~\eqref{eq:integrated-control-action} gives $\mathcal A_{\rm ctrl}=0$.

For a channel $\Phi$, define the worst-case target error
\begin{equation}
 \mathcal E_{\rm wc}(\Phi)=
 \sup_\rho\frac12\|\Phi(\rho)-P_0\|_1,
\label{eq:worst-case-target-error}
\end{equation}
where the supremum is over qutrit density operators.  For a physical evolution
$\Phi_t$, write
$E_{\rm wc}(t)=\mathcal E_{\rm wc}(\Phi_t)$.

\begin{proposition}[Ladder convergence]
\label{prop:ladder-convergence}
Every primitive and dead interval fixes $P_0$.  For the factorizing schedule,
the target error at integer periods is at most $2e^{-kT/2}$ for every integer
$k\ge0$, and for every wall time $t\ge0$ it satisfies
\begin{equation}
 E_{\rm wc}(t)\le2e^{T/2}
 \exp\!\left[-\frac{T}{2\mathcal P(T)}t\right].
\label{eq:ladder-all-time-envelope}
\end{equation}
Target error does not rise within a period.
\end{proposition}

The proof is given in Appendix~\ref{subsec:ladder-convergence-proof}.

In wall time $t$ measured from the declared strobe, the physical channel
agrees with $e^{[T/\mathcal P(T)]tL_\ell}$ at every strobe. Between strobes
it exhibits nonzero micromotion relative to this semigroup, while the target
error remains monotone.

\begin{theorem}[Ladder resource law]
\label{thm:ladder-resource-law}
Consider the two-label grammar with exclusive unit-amplitude access to
$A_\ell$ and $B_\ell$, minimum active-arc dwell $h>0$, and an identity dead
interval of duration $\sigma>0$ at each label change.  State injection and
the resource clock begin at the declared observation strobe after controller
initialization. Initialization and pre-running costs are excluded.
\begin{enumerate}
\item \emph{Exposure and resource lower bounds.}
For any measurable periodic protocol in this grammar, let $X_A(t)$ and
$X_B(t)$ be its cumulative exposures and let $E_{\rm wc}(t)$ be the worst-case
trace-distance error from $P_0$.  Then, for every $t\ge0$,
\begin{equation}
 E_{\rm wc}(t)\ge\frac12
 \exp\!\left[-\frac12\min\{X_A(t),X_B(t)\}\right].
\label{eq:ladder-exposure-obstruction}
\end{equation}
Let $T_{\rm wall}$ be the time of one complete period from the injection
strobe, with exposures $X_A$ and $X_B$. If $E_{\rm wc}(T_{\rm wall})\le\eps$
for $0<\eps<1/2$, then, with $s_\eps=2\log[1/(2\eps)]$,
\begin{equation}
 T_{\rm wall}\ge2\max\{h,s_\eps\}+2\sigma,\qquad
 X_A+X_B\ge2s_\eps,\qquad N_{\rm arc}\ge2,\qquad N_{\rm sw}\ge2.
\label{eq:ladder-resource-lower}
\end{equation}
\item \emph{Achieving schedule.}
For $0<\eps<1/2$, the identity in Theorem~\ref{thm:ladder-factorization}, with
\begin{equation}
 \widehat T_\eps=\max\!\left\{h,
 \left\lceil2\log\frac{2}{\eps}\right\rceil\right\}
\label{eq:ladder-resource-choice}
\end{equation}
achieves error at most $\eps$ after one period with wall time at most
$2\widehat T_\eps+2\sigma$, two active arcs, and two cyclic switches.
\item \emph{Matching leading cost.}
For fixed $h$ and $\sigma$, the upper and lower wall-time bounds have the
common leading term
\begin{equation}
 4\log(1/\eps)
\label{eq:ladder-leading-law}
\end{equation}
and differ by at most $8\log2+2$.
\end{enumerate}
\end{theorem}

The lower bound in Eq.~\eqref{eq:ladder-exposure-obstruction} is obtained
by applying the protocol to $(|0\rangle+|1\rangle)/\sqrt2$ and
$(|0\rangle+|2\rangle)/\sqrt2$. Their target coherences decay as
$e^{-X_A/2}$ and $e^{-X_B/2}$, so both exposures are necessary. Exclusivity
converts these exposure requirements into wall-time bounds, and the finite
factorization provides the matching upper construction. The exposure argument is detailed in
Sec.~\ref{subsec:ladder-exposure-proof}, and the upper error bound follows from
Proposition~\ref{prop:ladder-convergence}.

\begin{corollary}[Periodic exponent obstruction]
\label{cor:periodic-exponent-obstruction}
Let an admitted periodic protocol have period $\mathcal P$, per-period
exposures $s_A,s_B$, and $n_{\rm sw}$ dead switches.  If constants $K>0$ and
$\Gamma>0$ satisfy
$E_{\rm wc}(n\mathcal P)\le K e^{-\Gamma n\mathcal P}$ for every integer
$n\ge1$, then
\begin{equation}
 \Gamma\le\frac{\min\{s_A,s_B\}}{2\mathcal P}
 \le\frac{\mathcal P-n_{\rm sw}\sigma}{4\mathcal P}.
\label{eq:periodic-exponent-obstruction}
\end{equation}
The factorizing schedule has $s_A=s_B=T$ and $n_{\rm sw}=2$ and attains both
bounds.
\end{corollary}

\begin{proof}
After $n$ periods, Eq.~\eqref{eq:ladder-exposure-obstruction} and the assumed
envelope give
\[
 \frac12e^{-n\min\{s_A,s_B\}/2}
 \le E_{\rm wc}(n\mathcal P)\le Ke^{-\Gamma n\mathcal P}.
\]
Taking logarithms, dividing by $n$, and letting $n\to\infty$ proves the first
inequality.  Exclusivity, the unit-amplitude cap, and the dead intervals give
$s_A+s_B\le\mathcal P-n_{\rm sw}\sigma$, which proves the second.  For the
factorizing schedule, $\mathcal P=2T+2\sigma$. The bound in
Eq.~\eqref{eq:ladder-all-time-envelope} and the coherence witness in
Appendix~\ref{subsec:ladder-convergence-proof} give the common
sharp exponent $T/(2\mathcal P)$.
\end{proof}

\begin{remark}[Hardware transfer]
A target-stabilizer conjugation of $A_\ell$ cannot produce $B_\ell$, because it
preserves the jump range inside $\spann\{|0\rangle\}$ whereas $B_\ell$ has
range $\spann\{|1\rangle\}$.  Transferring Eq.~\eqref{eq:ladder-leading-law}
to a hardware presentation requires independent preservation of the exposure
witnesses, endpoint identity or actual monodromy, cold accounting, and the
same admissible comparator class.
\end{remark}

\subsection{Static and algebraic data for the cascade ladder}

\begin{proposition}[Cascade static certificate]
\label{prop:ladder-certificate}
The ladder has the positive-gauge data
\[
 K_A=K_B=\{I_3\},\qquad E_H=\{0\},\qquad \phi(A_\ell)=\phi(B_\ell)=\mu=3,\qquad s=m=2,\qquad \lambda_A=\lambda_B=1.
\]
There is no leakage from $P_0$.  Write $L_\ell^{\rm tr}$ for the matrix induced
by $L_\ell$ on the transverse quotient.  In the transverse basis of Proposition~\ref{prop:common-ground-word},
$A_\star:=L_\ell^{\rm tr}$ is
\begin{equation}
 A_\star=
 \begin{pmatrix}-1&1\\0&-1\end{pmatrix}
 \oplus\left(-\frac12I_4\right)\oplus(-I_2).
\label{eq:ladder-transverse-matrix}
\end{equation}
It admits the strict Lyapunov data
\begin{equation}
 M=\begin{pmatrix}1/2&1/4\\1/4&3/4\end{pmatrix}
 \oplus I_4\oplus\tfrac12I_2,\qquad m_-=\tfrac14,\qquad m_+=2,\qquad
 \nu=\tfrac1{16},\qquad q_0=\tfrac14,
\label{eq:ladder-lyapunov-data}
\end{equation}
with
\[
 A_\star^{\mathsf T}M+MA_\star=-I_8.
\]
The associated norm-conversion constants are
\[
 K_2:=\sqrt{m_+/m_-}=2\sqrt2,\qquad
 K_{\rm tr}:=2\sqrt3\,K_2=4\sqrt6.
\]
\end{proposition}

\begin{proof}
For $L_\ell=A_\ell+B_\ell$, direct calculation gives
\[
 L_\ell(P_1)=P_0-P_1,\qquad L_\ell(P_2)=P_1-P_2.
\]
The generator $L_\ell$ multiplies each of $X_{01},Y_{01},X_{02},Y_{02}$ by $-1/2$, and multiplies
$X_{12},Y_{12}$ by $-1$.  Thus $P_0$ is fixed, there is no leakage from the
target, and the quotient matrix is Eq.~\eqref{eq:ladder-transverse-matrix}.
Its characteristic polynomial is $(z+1)^4(z+1/2)^4$.  Direct multiplication
with Eq.~\eqref{eq:ladder-lyapunov-data} gives
\[
 A_\star^{\mathsf T}M+MA_\star=-I_8,\qquad -\bigl(A_\star^{\mathsf T}M+MA_\star+2\nu M\bigr)-q_0I_8
 =\frac34I_8-\frac18M\succ0,
\]
which verifies the stated Lyapunov margin.  The two orbit sets are singletons,
$E_H=0$, the weights are one, and minus the real superoperator trace gives
$\phi(A_\ell)=\phi(B_\ell)=3$.  Finally,
$K_2=\sqrt{m_+/m_-}=2\sqrt2$ and
$K_{\rm tr}=2\sqrt3\,K_2=4\sqrt6$.
\end{proof}

For the endpoint error estimate in Proposition~\ref{prop:rational-phase},
direct action on a system operator $X$ gives
\begin{equation}
 [A_\ell,B_\ell](X)=X_{22}(P_0-P_1),\qquad
 \|[A_\ell,B_\ell]\|_\diamond=2.
\label{eq:ladder-diamond-commutator}
\end{equation}
The amplified map extracts the $22$ block and tensors it with $P_0-P_1$.
The resulting output has trace norm at most twice the input trace norm, and a state supported on
level two attains the bound.

\subsection{Convergence and a rational observation phase}
\label{subsec:ladder-convergence-proof}

We first derive the semigroup convergence estimate, then quantify the effect
of rounding the observation phase.

Under $e^{tL_\ell}$, the populations satisfy
\begin{equation}
 p_2(t)=e^{-t}p_2(0),\qquad
 p_1(t)=e^{-t}[p_1(0)+tp_2(0)].
\label{eq:ladder-populations}
\end{equation}
Positivity gives
$\|[\rho_{01}(0),\rho_{02}(0)]\|_2\le1/2$.  Combining this with
Eq.~\eqref{eq:ladder-populations} yields
\begin{equation}
 \dtr(e^{tL_\ell}\rho,P_0)
 \le e^{-t}(1+t)+\frac12e^{-t/2}<2e^{-t/2}.
\label{eq:ladder-semigroup-error}
\end{equation}
The exact endpoint identity in Eq.~\eqref{eq:ladder-factorization} transfers
this semigroup bound to the physical schedule at each strobe. Every primitive and dead interval fixes the target,
so CPTP contractivity between strobes and
$\lfloor t/\mathcal P(T)\rfloor\ge t/\mathcal P(T)-1$ give
Eq.~\eqref{eq:ladder-all-time-envelope}.  The initial state
$(|0\rangle+|1\rangle)/\sqrt2$ gives a lower stroboscopic decay
$e^{-kT/2}/2$, so the displayed wall-time exponent is sharp for this
schedule.

\begin{proposition}[Rational observation phase]
\label{prop:rational-phase}
Fix $T>0$ and choose $\widehat y\in\mathbb Q\cap(0,T)$. Set
$\widehat x=T-\widehat y$ and
$\Delta_y=\widehat y-y(T)$, and define
$\widehat\Phi_T=e^{\widehat yA_\ell}e^{TB_\ell}
e^{\widehat xA_\ell}$.  Then
\begin{equation}
 \ddia\!\left(\widehat\Phi_T,e^{T(A_\ell+B_\ell)}\right)
 =Te^{-T}|e^{-\Delta_y}-1|
 \le Te^{-T}e^{|\Delta_y|}|\Delta_y|.
\label{eq:ladder-rounded-phase}
\end{equation}
Writing
\[
 r_{\widehat y}=e^{-\widehat y}(1-e^{-T}),\qquad
 K_{\widehat y,T}=\frac32+\frac{2r_{\widehat y}}{e\,Te^{-T}},
\]
the rounded period map remains GAS and obeys, for every integer $k\ge0$,
\begin{equation}
 \mathcal E_{\rm wc}(\widehat\Phi_T^k)
 \le K_{\widehat y,T}e^{-kT/2}.
\label{eq:rounded-stroboscopic-envelope}
\end{equation}
For $T\ge h$, let $\widehat\Psi_t$ denote the physical evolution from the
rounded observation strobe, so that
$\widehat\Psi_{k\mathcal P(T)}=\widehat\Phi_T^k$ for every integer $k\ge0$.
For every wall time $t\ge0$,
\begin{equation}
 \mathcal E_{\rm wc}(\widehat\Psi_t)
 \le K_{\widehat y,T}e^{T/2}
 \exp\!\left[-\frac{Tt}{2\mathcal P(T)}\right].
\label{eq:rounded-all-time-envelope}
\end{equation}
Rounding preserves the physical period, exposures, target, active-arc and
switch counts, and exponent. The endpoint equals the ideal channel exactly
when $\Delta_y=0$.
\end{proposition}

\begin{proof}
Only the level-$2$ to level-$1$ coefficient changes, and
Eq.~\eqref{eq:ladder-diamond-commutator} gives
Eq.~\eqref{eq:ladder-rounded-phase}.

For repeated periods, powers of the population quotient block
$\begin{psmallmatrix}e^{-T}&r_{\widehat y}\\0&e^{-T}\end{psmallmatrix}$
give an exponentially decaying population contribution with an additional
factor linear in $k$. Applying $ke^{-kT/2}\le2/(eT)$ to the $k$-dependent
term bounds the population contribution by a multiple of $e^{-kT/2}$.
Adding this bound to the coherence estimate used in
Eq.~\eqref{eq:ladder-semigroup-error} gives
Eq.~\eqref{eq:rounded-stroboscopic-envelope} with the stated constant
$K_{\widehat y,T}$. Contractivity within each period and
$\lfloor t/\mathcal P(T)\rfloor\ge t/\mathcal P(T)-1$ give
Eq.~\eqref{eq:rounded-all-time-envelope}.
\end{proof}

\subsection{Exposure lower bound}
\label{subsec:ladder-exposure-proof}

For the lower bound, the two target coherences retain exact dependence on
generator exposure under any admitted ladder protocol. The master equation gives
\begin{equation}
 \rho_{01}(t)=e^{-X_A(t)/2}\rho_{01}(0),\qquad
 \rho_{02}(t)=e^{-X_B(t)/2}\rho_{02}(0).
\label{eq:ladder-coherence-exposures}
\end{equation}
Applying the two target-coherence witnesses stated above proves
Eq.~\eqref{eq:ladder-exposure-obstruction}. If the error after one period
from the injection strobe is at most $\eps\in(0,1/2)$, each exposure in that
period is at least $s_\eps>0$. Both labels therefore occur. By the unit-amplitude cap and minimum
dwell, each label is active for at least $\max\{h,s_\eps\}$. Exclusivity adds
these two bounds, and the period has at least two active arcs and two cyclic
switches, whose dead intervals add at least $2\sigma$. This
proves Eq.~\eqref{eq:ladder-resource-lower} and, with the finite upper
construction, the sharp leading coefficient.

\bibliography{bibliography/Reference}

\end{document}